\documentclass[journal, onecolumn]{IEEEtran}
\usepackage{bm}

\usepackage{fullpage}  

\usepackage{url}
\usepackage[ruled,linesnumbered]{algorithm2e}

\SetKwProg{Init}{init}{}{}

\usepackage{multirow} 
\usepackage{hhline}  
\usepackage{verbatim}
\usepackage{graphicx}
\usepackage{amsmath,amssymb,amsfonts,graphicx,amsthm,mathtools,nicefrac}
\usepackage{lscape}
\usepackage{color}
\usepackage{authblk}
\usepackage{subfigure}
\usepackage{enumerate}
\usepackage{booktabs}

\usepackage{pifont}       
\usepackage{bbding}       
\usepackage{fontawesome}  
 
\allowdisplaybreaks
\newcommand\numberthis{\addtocounter{equation}{1}\tag{\theequation}}

\newtheorem{theorem}{Theorem}[section]
\newtheorem{assumption}{Assumption}[section]
\newtheorem{remark}{Remark}[section]

\numberwithin{equation}{section}

\newtheorem{lemma}[theorem]{Lemma}

\DeclareMathOperator{\rank}{\mathrm{rank}}
\DeclareMathOperator{\diag}{\mathrm{diag}}

\DeclareMathOperator{\trace}{trace}

\DeclareMathOperator{\blkdiag}{blkdiag} 

\def\lab{\left|}
\def\rab{\right|} 
\def\la {\left\langle}
\def\ra {\right\rangle} 
\def \lb{\left(}
\def \rb{\right)}

\newcommand{\matsnorm}[2]{\left\| #1\right\|_{{#2}}}

\newcommand{\fronorm}[1]{\ensuremath{\matsnorm{#1}{\footnotesize{\mathsf{F}}}}}
\newcommand{\opnorm}[1]{\ensuremath{\matsnorm{#1}{}}}

\newcommand{\twonorm}[1]{\ensuremath{\matsnorm{#1}{\footnotesize{2}}}}

\newcommand{\bfm}[1]{\bm{#1}}

\renewcommand{\Pr}[2][]{\mathbb{P}_{#1} \left\{ #2 \rule{0mm}{3mm}\right\}}

\newcommand{\E}[2][]{\mathbb{E}_{#1} \left\{ #2 \rule{0mm}{3mm}\right\}}

\def\va{\bfm a}   \def\mA{\bfm A}  
\def\vb{\bfm b}   \def\mB{\bfm B}  
     \def\C{\mathbb{C}}
   \def\mD{\bfm D}  
\def\ve{\bfm e}     
    
\def\vg{\bfm g}   \def\mG{\bfm G}  
\def\vh{\bfm h}     
   \def\mI{\bfm I}

   \def\mL{\bfm L}

   \def\mO{\bfm O}

   \def\mR{\bfm R}  \def\R{\mathbb{R}}

   \def\mU{\bfm U}  
\def\vv{\bfm v}   \def\mV{\bfm V}  
   \def\mW{\bfm W}  
\def\vx{\bfm x}   \def\mX{\bfm X}  
\def\vy{\bfm y}   \def\mY{\bfm Y}  
\def\vz{\bfm z}   \def\mZ{\bfm Z}

\def\calA{{\cal  A}}

\def\calD{{\cal  D}} 
\def\calE{{\cal  E}} 
 
\def\calG{{\cal  G}} 
\def\calH{{\cal  H}} 
\def\calI{{\cal  I}}

\def\calO{{\cal  O}} 
\def\calP{{\cal  P}}

\def\calS{{\cal  S}}

\newcommand{\bfsym}[1]{\bm{#1}}

             \def\bSigma{\bfsym \Sigma}

\def\mPhi {\bfsym{\Phi}}
\def\mPsi {\bfsym{\Psi}}

\def \calGT {\calG^\ast}

\def \tran {\mathsf{T}}
\def \tranH{\mathsf{H}}

\def \bzero{\bm 0}

\usepackage{bbm}

\def \mDelta {{\bfsym {\Delta}}}

\def \tG {\widetilde{\calG}}

\def \tcalI {\widetilde{\calI}}

\def \mUpsilon {\bfsym{\Upsilon}}

\begin{document}
	
\title{Fast and Provable Simultaneous Blind Super-Resolution and Demixing for Point Source Signals: Scaled Gradient Descent without Regularization }
\author[1]{Jinchi Chen}
\affil[1]{School of Mathematics, East China University of Science and Technology, Shanghai, China.\vspace{.15cm}}
\date{\today}

\maketitle

\begin{abstract}
	We address the problem of simultaneously recovering a sequence of point source signals from observations limited to the low-frequency end of the spectrum of their summed convolution, where the point spread functions (PSFs) are unknown. By exploiting the low-dimensional structures of the signals and PSFs, we formulate this as a low-rank matrix demixing problem. To solve this, we develop a scaled gradient descent method without balancing regularization. We establish theoretical guarantees under mild conditions, demonstrating that our method, with spectral initialization, converges to the ground truth at a linear rate, independent of the condition number of the underlying data matrices. Numerical experiments indicate that our approach is competitive with existing convex methods in terms of both recovery accuracy and computational efficiency.
\end{abstract}

\section{Introduction}

\subsection{Problem Formulation}
The problem of blind super-resolution is about  estimating high-resolution information of a signal from its measurements at the lower end of spectrum when the point spread functions are unknown. This technique is widely used in various fields, including seismic data analysis \cite{margrave2011gabor}, nuclear magnetic resonance spectroscopy \cite{qu2015accelerated}, and 3D single-molecule microscopy \cite{quirin2012optimal}. In this paper, we consider a more general and difficult problem in which we have to extract multiple point sources from one low-frequency observation. Such problem arises in a range of applications such as integrated sensing and communications \cite{liu2020joint, chafii2023twelve}, over the air computation \cite{nazer2007computation, goldenbaum2014nomographic} and internet-of-things \cite{wunder20145gnow}.

More specifically, consider a series of point source signal of the form:
\begin{align*}
	x_k(t) = \sum_{p=1}^{r} d_{k,p} \delta(t-\tau_{k,p}), k=1,\cdots, K,
\end{align*}
where $K$ means the number of signals or users, $\delta(\cdot)$ represents the Dirac function, $r$ indicates the number of spikes of the $k$-th signal, and $\{\tau_{k,p}\}_{p= 1}^{r}\subset [0,1)$ and $\{d_{k,p}\}_{p =1}^{r}\subset \C$ represent the locations and amplitudes of the point source signals, respectively. 

Let $y(t)$ be the summation of point source signals convolved with unknown point spread functions, given by
\begin{align}
	\label{eq: receive signal}
	y(t) = \sum_{k = 1}^K\lb \sum_{p=1}^{r}d_{k,p} \delta(t-\tau_{k,p})\ast g_{k,p}(t) \rb = \sum_{k=1}^{K}   \sum_{p = 1}^{r} d_{k,p} g_{k, p}(t-\tau_{k,p}).
\end{align}
By taking the Fourier transform of~\eqref{eq: receive signal} 
\begin{align*}
	\hat{y}(f) &= \sum_{k=1}^{K}\sum_{p=1}^{r} d_{k, p} \int g_{k,p}(t-\tau_{k,p}) e^{-2\imath \pi f t}dt=\sum_{k=1}^{K}\sum_{p=1}^{r} d_{k, p} \hat{g}_{k,p}(f)e^{-2\imath \pi f \tau_{k,p}}dt,
\end{align*}
where $\hat{g}_{k,p}(f)$ is the Fourier transform of $g_{k,p}(t)$. Subsequently sampling, we obtain that
\begin{align}\label{eq: observation}
	\hat{y}[j] = \sum_{k=1}^{K} \sum_{p = 1}^{r} d_{k,p} e^{-\imath 2\pi (j-1)\tau_{k,p}} \hat{g}_{k,p}[j], \quad j=0,\cdots, n-1.
\end{align}
where $\vg_{k, p} := \begin{bmatrix}
	\hat{g}_{k,p}[0] &\cdots & \hat{g}_{k,p}[n-1] 
\end{bmatrix}^\tran\in\C^n$
are unknown. The goal is to jointly recover both $\{d_{k,p}, \tau_{k,p}\}$ and $\{\vg_{k,p}\}$ from~\eqref{eq: observation}. We refer to this issue as simultaneous blind super-resolution and demixing. Specifically, when there is only one user ($K=1$), the measurement model in \eqref{eq: observation} simplifies to the problem of blind super-resolution, as discussed in \cite{chi2016guaranteed,yang2016super,li2019atomic,chen2022vectorized}. Furthermore, if there are two types of signals ($K=2$), with one related to communication and the other associated with radar, the problem described in \eqref{eq: observation} can be considered a case of dual-blind deconvolution, as examined in \cite{vargas2023dual}.

Since the number of unknowns $\calO(Knr)$ in \eqref{eq: observation} is larger than the number of samples, simultaneous blind super-resolution and demixing is an ill-posed problem. To alleviate this issue, we follow the same route as that in  \cite{ahmed2013blind, chi2016guaranteed, yang2016super,chen2022vectorized}  and assume that the set of vectors $\{\vg_{k,p}\}_{p=1}^{r}$ associated with the $k$-th unknown point spread function belong to a known low dimensional subspace spanned by the columns of $\mB_k\in\C^{n\times s}$ with $s<n$, i.e., 
\begin{align*}
	\vg_{k,p} = \mB_{k}\vh_{k,p}, \quad p=1,\cdots, r,
\end{align*}
where $\vh_{k,p}\in\C^{s}$ denotes an unknown coefficient vector. Then the measurements can be rewritten as follows:
\begin{align}
	\label{eq sampling model}
	y[j] &=\sum_{k=1}^{K}\la \vb_{k,j}\ve_j^\tran, \sum_{p= 1}^{r} d_{k,p} \vh_{k,p}\va_{\tau_{k,p}}^\tran\ra:= \sum_{k=1}^{K} \la \vb_{k,j}\ve_j^\tran, \mX_{k,\natural} \ra,\quad j=0,\cdots, n-1,
\end{align}
where the inner product of two matrices $\mA,\mB$ is defined as $\la\mA,\mB\ra=\trace(\mA^\tranH\mB)$, $\ve_j\in\R^n$ represents the $j$-th standard basis vector in $\R^n$, $\vb_{k,j}\in\C^{s}$ denotes the $j$-th column of $\mB_k^\tranH$, and $\va_{\tau}\in\C^n$ stands as the steering vector, defined as
\begin{align*}
	\begin{bmatrix}
		1 & e^{-\imath2\pi \cdot 1\cdot \tau} &\cdots & e^{-\imath2\pi \cdot (n-1)\cdot \tau} 
	\end{bmatrix}^\tran\in\C^{n}.
\end{align*}
Let $\calA_k:\C^{s\times n}\rightarrow\C^n$ be a linear operator such that 
\begin{align}
	\label{A}
	\calA_k(\mX)[j] = \la \vb_{k,j}\ve_j^\tran, \mX\ra.
\end{align}
Then the measurement model \eqref{eq sampling model} can be rewritten as 
\begin{align}
	\label{eq measurement}
	\vy = \sum_{k=1}^{K}\calA_k(\mX_{k,\natural}),
\end{align}
where $\vy = \begin{bmatrix}
	y[0] &\cdots & y[n-1]
\end{bmatrix}^\tran$. 

Based on the aforementioned reformulation of simultaneous blind super-resolution and demixing under the subspace assumption, it is evident that the key is to recover ${\mX_{k,\natural}}$ from the linear measurement vector $\vy$. Once the data matrix $\mX_{k,\natural}$ is reconstructed, the locations  $\{\tau_{k,p}\}$ of the $k$-th signal can be determined using spatial smoothing MUSIC \cite{evans1981high,evans1982application,yang2019source,chen2022vectorized}. With all the location information available, $\{d_{k,p}\}$ and $\{\vh_{k,p}\}$ can be recovered by solving an overdetermined linear system. In the following sections, we focus on the problem of jointly recovering $\{\mX_{k, \natural}\}$ from their mixed linear measurements.

\subsection{Related Work and Contributions}
The convex optimization method for addressing the simultaneous blind super-resolution and demixing problem has been extensively studied in recent years. In \cite{vargas2023dual,jacome2024multi}, the authors build upon the approach initially developed in \cite{yang2016super} for blind super-resolution and introduce an atomic norm minimization (ANM) method for the simultaneous blind super-resolution and demixing problem specifically with $K=2$. More recently, \cite{daei2024timely} extends the ANM approach to handle arbitrary $K$ values, providing theoretical guarantees. 
Apart from ANM, another line of research exploits the low-rank structure of the target matrix associated with the signals. Specifically, leveraging the vectorized Hankel lift (VHL) technique introduced in \cite{chen2022vectorized}, \cite{monsalve2023beurling} proposes a nuclear norm minimization approach for dual-blind deconvolution with $K=2$, which is subsequently generalized to arbitrary $K$ in \cite{wang2024simultaneous}.
However, these convex approaches struggle with large-scale problems. This challenge motivates us to design efficient and provable algorithms for the simultaneous blind super-resolution and demixing problem.

When $K=1$, the simultaneous blind super-resolution and demixing problem simplifies to the case of blind super-resolution. Previous approaches, including atomic norm minimization \cite{chi2016guaranteed,li2019atomic} and nuclear norm minimization \cite{chen2022vectorized}, have been thoroughly studied for addressing blind super-resolution. Recently, several studies have focused on developing and analyzing efficient algorithms with provable guarantees for blind super-resolution \cite{mao2022blind,zhu2021blind}. Our research extends the concept of blind super-resolution (with $K=1$)  to a multi-user framework (with $K>1$). However, our approach differs significantly from existing methods. The PGD-VHL method proposed in \cite{mao2022blind} relies on additional projections to maintain incoherence properties and employs explicit regularization to balance the norms of two low-rank factors. In contrast, our method eliminates the need for such projections or balancing terms. 
Concurrently, a scaled gradient descent method for blind super-resolution has been proposed with theoretical guarantees (Jinsheng Li, Wei Cui, and Xu Zhang, personal communication). In contrast, our study addresses a more challenging problem, focusing on the simultaneous blind super-resolution and demixing problem. The key to our approach lies in constructing a suitable basin of attraction and understanding the properties of the linear sensing operator when constrained to a block-diagonal subspace. Consequently, our findings are unique and not covered by existing research.

This work is also related to the problem of low-rank matrix demixing, which involves the joint recovery of multiple low rank matrices from a single measurement vector. Specifically, McCoy and Tropp \cite{mccoy2013achievable} study the compressive demixing model represented by $\vy = \calA(\sum_{k=1}^{K}\mX_k)$ and propose a nuclear norm minimization approach to recover $\{\mX_k\}$. Strohmer and Wei \cite{strohmer2019painless} introduce a provable non-convex algorithm for the demixing problem based on the Gaussian measurement model. Additionally, both convex \cite{jung2017blind, ling2017blind} and non-convex \cite{ling2019regularized,dong2018nonconvex} optimization methods have been examined for rank-one matrix demixing from rank-one observations. However, due to the distinct structures of the target matrices and sensing operators in our work, the theoretical guarantees from these studies cannot be directly extended to our model.

In recent years, substantial advancements have been made in the development of scaled or preconditioned algorithms for low-rank matrix recovery. Essentially, these scaled algorithms are designed to function similarly to a (Riemannian) gradient descent method, incorporating appropriately chosen preconditioning techniques. Empirical and theoretical studies have shown that such scaled methods are more robust to ill-conditioned low-rank matrix recovery problems while maintaining a comparable computational cost to gradient descent methods. Furthermore, scaled gradient descent benefits from an implicit regularization property, eliminating the need to explicitly balance the norm of low-rank factors, as required in the traditional literature (see \cite{tu2016low,zheng2016convergence,chi2019nonconvex} and references therein). The empirical performance of scaled (alternating) gradient descent for the low-rank matrix completion problem has been illustrated in \cite{mishra2012riemannian, mishra2016riemannian, tanner2016low}. Recently, theoretical convergence guarantees have also been established. Tong et al. \cite{tong2021accelerating} demonstrated that the scaled gradient descent method for the low-rank matrix sensing problem achieves a linear convergence rate given spectral initialization. Zhang et al. \cite{zhang2023preconditioned} extended the results of \cite{tong2021accelerating} to the over-parameterized case. The convergence guarantee for scaled gradient descent with small random initialization is provided by Xu et al. \cite{xu2023power}, while Cheng and Zhao further generalized the findings of Zhang et al. to the asymmetric case. Despite some general similarities, there is minimal overlap between the aforementioned studies and our framework.

The main contributions of this work are summarized as follows:
\begin{itemize}
	\item We address the simultaneous blind super-resolution and demixing problem by formulating it as a low-rank matrix demixing problem. Leveraging low-rank structures, we propose an implementable algorithm named the Scaled Gradient Descent method (Scaled-GD), which operates without explicit balancing regularization. To the best of our knowledge, this is the first provable and computationally efficient algorithm for simultaneous blind super-resolution and demixing.
	\item We establish the convergence rate of Scaled-GD, demonstrating that it linearly converges to the ground truth from spectral initialization, provided there is no noise and the sample complexity is of the order $\calO(K^2s^2r^2\kappa^2\mu_1\mu_0\log^2(sn))$. Additionally, the convergence rate of Scaled-GD is independent of the condition number of the ground truth, highlighting the efficiency of the proposed method.
	\item We empirically validate the recovery performance of Scaled-GD. Numerical experiments show that Scaled-GD is competitive with convex recovery methods such as ANM and VHL in terms of recovery ability, but is much more efficient.
\end{itemize}

\subsection{Organization and Notation}
This paper is organized as follows. Section \ref{sec methodology} presents the recovery procedure and the Scaled-GD algorithm. The theoretical guarantees for the proposed method are provided in Section \ref{sec main results}. Numerical results are discussed in Section \ref{sec numerical}. The proofs detailed of our main results are presented in Section \ref{sec proof main results}. We conclude this work in Section \ref{sec conclusion}. Appendix includes the proofs of immediate results.

Throughout this work, vector, matrix and operator are represented by bold lowercase letters, bold uppercase letters and calligraphic letters, respectively. Given a series of matrices $\{\mL_k\}_{k=1,\cdots, K}$, let $\mL=\blkdiag(\mL_1,\cdots, \mL_K)$ be the block diagonal matrix, where $\blkdiag$ is the MATLAB function to construct block diagonal matrix. Given block diagonal matrices $\mV$ and $\mW$, the scaled norm with respect to $\mL,\mR$ is defined as follows:
\begin{align}
	\label{def ds}
	d_{\calS}(\mV,\mW) = \sqrt{\fronorm{\mV(\mR^\tranH\mR)^{\frac{1}{2}}}^2 + \fronorm{\mW(\mL^\tranH\mL)^{\frac{1}{2}}}^2},
\end{align}
and its dual norm is given by
\begin{align}
	\label{def dual ds}
	d_{\calS^\ast}(\mV,\mW) = \sqrt{\fronorm{\mV(\mR^\tranH\mR)^{-\frac{1}{2}}}^2 + \fronorm{\mW(\mL^\tranH\mL)^{-\frac{1}{2}}}^2}.
\end{align}

For any matrix $\mX$, we  let $\opnorm{\mX}, \fronorm{\mX}$ denote spectral norm and  Frobenius norm of $\mX$. Let $\sigma_i(\mZ)$ is the $i$-th singular value of $\mZ_{k,\natural}$.  We let $f(n)=\calO(g(n))$ to mean $|f(n)|\leq C|g(n)|$ for some constant $C$ when $n$ is large enough. The term "with high probability" indicates that the event occurs with probability at least $1-(sn)^{-\gamma}$, where $\gamma$ is an absolute constant whose value may vary from line to line.

\section{Algorithm}
\label{sec methodology}
Let $\vx_{k,i}\in\C^{s}$ be the $i$-th column of $\mX_{k,\natural}$ such that 
\begin{align*}
	\mX_{k,\natural} = \begin{bmatrix}
		\vx_{k,1} &\cdots & \vx_{k,n}
	\end{bmatrix}\in\C^{s\times n},
\end{align*}
and $\calH$ be the vectorized Hankel lift operator which maps a $s\times n$ matrix into an $sn_1\times n_2$ matrix, 
\begin{align*}
	\calH(\mX_{k,\natural}) = \begin{bmatrix}
		\vx_{k,1} & \vx_{k,2} &\cdots & \vx_{k,n_2}\\
		\vx_{k,2} &\vx_{k,3} &\cdots & \vx_{k,n_2+1}\\
		\vdots      & \vdots     & \ddots & \vdots\\
		\vx_{k,n_1} &\vx_{k,n_1+1} &\cdots & \vx_{k,n}\\
	\end{bmatrix},
\end{align*}
where $n_1 + n_2 = n+1$. It has been shown that $\mZ_{k,\natural}:=\calH(\mX_{k,\natural})$ is a rank-$r$ matrix \cite{chen2022vectorized}. Let $\calG=\calH\calD^{-1}:\C^{s\times n}\rightarrow \C^{sn_\times n_2}$ be the weighted vectorized Hankel lift operator, where the operator $\calD:\C^{s\times n}\rightarrow \C^{s\times n}$ is defined as 
\begin{align*}
	\calD(\mX_{k,\natural}) = \begin{bmatrix}
		\sqrt{w_1}\vx_{k,1} &\cdots & \sqrt{w_{n}}\vx_{k,n}
	\end{bmatrix},
\end{align*}
and the scalar $w_i$ is the number of entries
on the $i$-th antidiagonal of an $n_1\times n_2$ matrix. It can be seen that the target matrix $\mZ_{k,\natural}$ obeys that $(\calI - \calG\calG^\ast)(\mZ_{k,\natural})=\bzero$ and the measurements in \eqref{eq measurement} can be rewritten as 
\begin{align*}
	\vy = \sum_{k=1}^{K}\calA_k\calD^{-1}\calG^\ast (\mZ_{k,\natural}).
\end{align*}
By exploring the low rank structure of $\mZ_{k,\natural}$, it is convenient to consider the following recovery procedure
\begin{align}
	\label{prob 1}
	\min_{\{\mZ_k\}_{k=1}^K}\frac{1}{2} \twonorm{\mD\vy - \sum_{k=1}^{K}\calA_k\calG^\ast(\mZ_k) }^2\quad \text{s.t.}\quad \rank(\mZ_k) = r, ~(\calI-\calG\calGT)(\mZ_k)=\bzero \text{ for }k=1,\cdots, K,
\end{align}
where $\mD = \diag(\sqrt{w_1}, \cdots, \sqrt{w_n})$. Furthermore, to encode the rank-$r$ constraint in \eqref{prob 1}, we parameterize the vectorized Hankel matrix $\mZ_k\in\C^{sn_1\times n_2}$ as a product of two matrices $\mL_k\in\C^{sn_1\times r}$ and $\mR_k\in\C^{n_2\times r}$.  Therefore, we consider the following penalized version of \eqref{prob 1} for estimating the factorized matrices:
\begin{align}
	\label{optimization}
	\min_{\{\mL_k, \mR_k\}_{k=1}^K} \left\{f\lb \{\mL_k, \mR_k\}_{k=1}^K\rb:=\frac{1}{2}\twonorm{\mD\vy - \sum_{k=1}^{K} \calA_k\calGT(\mL_k\mR_k^\tranH) }^2 + \frac{1}{2} \sum_{k=1}^{K}\fronorm{\lb \calI - \calG\calGT\rb(\mL_k\mR_k^\tranH)}^2\right\}.
\end{align}
Although this represents a highly non-convex optimization problem over $\{\mL_k,\mR_k\}_{k=1,\cdots, K}$, recent breakthroughs have demonstrated that simple gradient descent, when combined with suitable initialization, can provably converge to the optimal minimum under mild statistical assumptions.

However, the gradient descent method faces two major challenges in the application of simultaneous blind super-resolution and demixing. Firstly, the data matrices in application may be ill-conditioned, which can significantly degrade the convergence rate of the gradient descent method. Secondly, it is common to introduce explicit regularization to balance the norms of 
$\mL_k$ and $\mR_k$ when applying the gradient descent method. While this regularization is often necessary for theoretical analysis, it is not required if a balanced initialization is provided.

Inspired by recent development in \cite{tong2021accelerating,zhang2023preconditioned}, we propose a scaled gradient descent method for the simultaneous blind super-resolution and demixing problem, which is summarized in Algorithm \ref{alg: SGD}. 
More specifically, the initialization is obtained by the spectral method from Line 1 to Line 5. The iterative update of $\{\mL_k, \mR_k\}$ are provided in Line 7 and 8, where the gradients are computed with respect to Wirtinger calculus given by
\begin{align*}
	\nabla f(\mL_k) &= \lb \calG\calA_k^\ast \lb \sum_{\ell=1}^{K} \calA_\ell\calGT\lb \mL_\ell\mR_\ell^\tranH - \mZ_{\ell, \natural}\rb \rb \rb \mR_k + \lb \lb \calI - \calG\calGT\rb(\mL_k\mR_k^\tranH)\rb\mR_k,\\
	\nabla f(\mR_k) &= \lb \calG\calA_k^\ast \lb \sum_{\ell=1}^{K} \calA_\ell\calGT\lb \mL_\ell\mR_\ell^\tranH - \mZ_{\ell, \natural}\rb \rb \rb^\tranH \mL_k + \lb \lb \calI - \calG\calGT\rb(\mL_k\mR_k^\tranH)\rb^\tranH\mL_k.
\end{align*}
The computational complexity of $\nabla f(\mL_k)$ (and similar for $\nabla f(\mR_k)$) is about $\calO(srn\log(n))$ \cite{mao2022blind}. Moreover, computing $\mR_{k}(\mR_k^\tranH\mR_k)^{-1}$ costs $\calO(r^2n)$. Therefore, the total computational cost in each iteration is about $\calO(Ksrn\log n+Kr^2n)$.
\begin{algorithm}[ht!]
	\caption{Scaled--GD for Simultaneous Blind Super-Resolution and Demixing}
	\label{alg: SGD}
	\tcp{Initialization}
	\For{$k=1,\cdots, K$}{
		\tcp{Fully parallel}
		$\mZ_{k,0} = \calP_r\calH\calA^\ast_k(\vy)$\;
		Let $\mZ_{k,0}=\mU_{k,0}\bSigma_{k,0}\mV_{k,0}^\tranH$\;
		$\mL_{k,0} \leftarrow \mU_{k,0}\bSigma_{k,0}^{\frac{1}{2}},~ \mR_{k,0} \leftarrow \mV_{k,0}\bSigma_{k,0}^{\frac{1}{2}}$\;
	}
	
	\For{$t=0,2,\cdots, T-1$}{
		\tcp{Fully parallel}
		\For{$k=1,\cdots, K$}{
			$\mL_{k, t+1} =\mL_{k,t} - \eta_t \nabla f(\mL_{k,t})(\mR_{k,t}^\tranH\mR_{k,t})^{-1}$\;
			$\mR_{k, t+1} =\mR_{k,t} - \eta_t \nabla f(\mR_{k,t})(\mL_{k,t}^\tranH\mL_{k,t})^{-1}$\;
		}
	}
	
\end{algorithm}

As aforemetioned, given the recovered data matrix $\mX_{k, \natural}$, the locations $\{\tau_{k,p}\}_{p=1,\cdots, r}$ can be extracted by spatial smoothing MUSIC algorithm \cite{evans1981high,evans1982application,yang2019source,chen2022vectorized}. Additionally, given these locations, the coefficients $\{d_{k,p}, \vh_{k,p}\}_{k=1,\cdots, K; p=1,\cdots, r}$ satisfies the following equation:
\begin{align*}
	y[j] &=\sum_{k=1}^{K}\sum_{\ell = 1}^{r_k}  \la \vb_{k,j}\ve_j^\tran, d_{k,\ell} \vh_{k,\ell}\va_{\tau_{k,\ell}}^\tran\ra
	=\begin{bmatrix}
		(\va_{\tau_{1,1}}^\tran \ve_j) \vb_{1,1}^\tranH&\cdots &  (\va_{\tau_{k,\ell}}^\tran \ve_j) \vb_{k,\ell}^\tranH&\cdots & (\va_{\tau_{K,r}}^\tran \ve_j) \vb_{K,r}^\tranH
	\end{bmatrix}\begin{bmatrix}
		d_{1,1} \vh_{1,1}\\
		\vdots\\
		d_{k,\ell} \vh_{k,\ell}\\
		\vdots\\
		d_{K,r} \vh_{K,r}\\
	\end{bmatrix}.
\end{align*}
Since the number of measurements $(\calO(K^2s^2r^2\kappa^2\mu_0\mu_1\log^2(sn)))$ is larger than $Ksr$, one can estimate $\{d_{k,p} h_{k,p}\}$ by solving an over-determined linear system.

\section{Main Results}
\label{sec main results}
In this section, we present a theoretical guarantee for the proposed approach. To achieve this, we introduce two standard assumptions.
\begin{assumption}[$\mu_0$-incoherence]
	\label{assumption 1}
	Suppose that the columns $\{\vb_{k,i}\}$ of $\mB_k^\tranH$ for $k=1,\cdots, K$, are i.i.d sampled from the a distribution $F$, which satisfies the following conditions for $\vb\sim F$:
	\begin{align*}
		\E{\vb} &= \bzero,\\ 
		\E{\vb\vb^\tranH} &=\mI,\\
		\max_{1\leq p\leq s}|\vb[p]| &\leq \sqrt{\mu_0}. 
	\end{align*}
\end{assumption}
\begin{remark}
This assumption was first introduced in RIPless compressed sensing \cite{candes2011probabilistic} and has also been adopted in blind super-resolution \cite{chi2016guaranteed,yang2016super,li2019atomic,chen2022vectorized} and simultaneous blind super-resolution and demixing \cite{vargas2023dual,jacome2024multi,wang2024simultaneous,daei2024timely}. Assumption \ref{assumption 1} holds for many random ensembles. For instance, when the components of $\vb$ are random variables taking values $-1$ or $1$ with equal probability, it implies that $\mu_0=1$. Additionally, this assumption is satisfied when the entries of $\vb$ are independently and identically sampled from the uniform distribution over $[-\sqrt{3}, \sqrt{3}]$. 
\end{remark}

\begin{assumption}[$\mu_1$-incoherence]
	\label{assumption 0}
	Let $\mZ_{k,\natural} = \mU_{k,\natural} \bSigma_{k,\natural}\mV_{k,\natural}^\tranH$ be the singular value decomposition of $\mZ_{k,\natural}$, where $\mU_{k,\natural}\in\C^{sn_1\times r}, \mV_{k,\natural}\in\C^{n_2\times r}$. Let $\mU_{k,\natural, j} = \mU_{k,\natural}[js:(j+1)s-1,:]\in\C^{s\times r}$ be the $j$-th block of $\mU_{k,\natural}$ for $j=0,\cdots, n_1-1$. Suppose that for all $k=1,\cdots, K$, the matrix $\mZ_{k,\natural} $ obeys the following conditions:
	\begin{align*}
		\max_{0\leq j\leq n_1-1} \fronorm{\mU_{k,\natural, j} }^2 \leq \frac{\mu_1r}{n}\text{ and } \max_{0\leq \ell\leq n_2-1} \twonorm{\ve_\ell^\tran\mV_{k,\natural}}^2\leq \frac{\mu_1r}{n}
	\end{align*}
	for some positive constant $\mu_1$.
\end{assumption}
\begin{remark}
Assumption \ref{assumption 0} is the same as the incoherence property made in low rank matrix completion \cite{candes2009} and has been used in blind super-resolution \cite{chen2022vectorized,mao2022blind}. This assumption is satisfied when the minimum wrap-up distance between the locations of point sources is larger than $\frac{2}{n}$ \cite{moitra2015super}.
\end{remark}

Now we present our main result, whose proof is deferred to Section \ref{proof iteration}.
\begin{theorem}
	\label{theorem: iteration}
	Suppose that Assumption \ref{assumption 0} and Assumption \ref{assumption 1} hold. Let $\eta_t\leq\frac{1}{20}$. If the number of measurements satisfies that $n\geq C_{\gamma} K^2s^2r^2 \kappa^2\mu_0 \mu_1 \log^2(sn)$, then with probability at least $1-(sn)^{-\gamma}$, the sequence $\{\mL_{k,t}, \mR_{k,t}\}_{k=1,\cdots, K}$ returned by Algorithm \ref{alg: SGD} obeys that
	\begin{align}
		\label{variable error}
		\sum_{k=1}^{K}\fronorm{\mL_{k,t}\mR_{k,t}^\tranH  - \mZ_{k,\natural}}^2 &\leq \lb 1- \frac{\eta_t}{90}\rb^t \cdot \frac{\delta^2\sigma_0^2}{K},
	\end{align}
where $\sigma_0= \sqrt{\sum_{k=1}^{K}\sigma_r^2(\mZ_{k, \natural}) }$ and $\kappa =\frac{\max_k \sigma_1(\mZ_{k, \natural})}{\min_k \sigma_r(\mZ_{k, \natural})}$. 

\end{theorem}

\begin{remark}
Theorem \ref{theorem: iteration} implies that the sample complexity depends quadratically on $s,r$ and $K$. The $r^2$ dependence arises from the initialization process. The convergence rate is established based on the Frobenius norm, whereas the initial error is bounded by the spectral norm. An additional $r$ factor is introduced to establish the relationship between these two norms. More details can be found in Lemma \ref{theorem: initial}. The quadratic dependence of $s$ and $K$ is due to the construction of the initial guess, which is designed to ensure that the initial error is of the order $1/\sqrt{K\mu_0 s}$ (see Lemma \ref{bound of nabla g} and Lemma \ref{upper bound of h}). This careful construction is essential for achieving the desired error bounds. However, we empirically observe a linear relationship between $n$ and $K$. Therefore, we acknowledge that the quadratic sample complexity is likely an artifact of our proof.
\end{remark}

\begin{remark}
	The sample complexities of the nuclear norm minimization method \cite{wang2024simultaneous} and the ANM method \cite{daei2024timely} are approximately $\calO(K^2sr\log^4(sn))$ and $\calO(Ksr\log^2(ns))$, respectively. Compared to these results, our sample complexity is sub-optimal with respect to 
	$s,r$ and $K$. However, as previously mentioned, our method is more efficient than these convex relaxation methods.  Moreover, it is worth noting that our results are established based on slightly milder assumptions than those in \cite{daei2024timely}. The theoretical performance guarantee in \cite{daei2024timely} relies on an additional assumption requiring the randomness of the coefficient vectors $\vh_{k,p}$. In contrast, the theoretical performance of Scaled-GD is independent of such assumptions.
\end{remark}

\begin{remark}
Theorem \ref{theorem: initial} establishes that Scaled-GD converges linearly to the sequence ${\mZ_{k,\natural}}$. To achieve an accuracy of $\varepsilon$, defined as $\sum_{k=1}^{K}\fronorm{\mL_{k,t}\mR_{k,t}^\tranH- \mZ_{k,\natural}}^2\leq \varepsilon\sigma_0^2$, Scaled-GD requires at most $T = \mathcal{O}(\log(1/\varepsilon))$ iterations, which is notable for its independence from the condition number $\kappa$. Given the iterates ${\mL_{k,T}, \mR_{k,T}}$, the estimate of $\mX_{k, \natural}$ can be obtained by applying $\calD^{-1}\calGT(\mL_{k,T}\mR_{k,T}^\ast)$.
\end{remark}
The proof of Theorem \ref{theorem: iteration} follows the primary ideas outlined in \cite{zhang2023preconditioned} and is deferred to Section \ref{sec proof main results}. Fundamentally, we establish the Lipschitz smoothness and the Polyak-Lojasiewicz inequality for the objective function in \eqref{optimization} with respect to a scaled norm. This framework enables us to demonstrate a linear convergence rate through induction. However, the specifics of the proof are intricate and notably distinct, preventing us from addressing them using existing methodologies. Our work diverges from the closely related study by \cite{zhang2023preconditioned} in two significant aspects:
\begin{itemize}
	\item We consider a more complex problem where $K$ structured matrices must be jointly recovered from a single measurement vector. To tackle this issue, we formulate both the target matrices and the sensing operator as block diagonal matrices. Furthermore, the measurement matrices are rank-$1$ matrices, which leads to the failure of the restricted isometry property (RIP) for the sensing operator. The primary challenge is to understand the properties of the sensing operator when restricted to block diagonal matrices. As a result, the convergence analysis presented in \cite{zhang2023preconditioned} cannot be easily generalized to our model.
	
	\item The construction of the basin of contraction is crucial for understanding non-convex optimization and varies across different problems. As a result, the initial guess slightly differs from that in \cite{zhang2023preconditioned}. Moreover, to ensure the initialization is close to the ground truth, we must carefully address the effect of inter-coherence between different target matrices. Consequently, additional effort is required to achieve the desired guarantee.

\end{itemize}

\section{Numerical experiments}
\label{sec numerical}
In this section, we conduct a series of numerical experiments to evaluate the performance of proposed method in simultaneous blind super-resolution and demixing. The numerical simulations are executed from MATLAB R2022b on a macOS machine with multi-core Intel CPU at 2.3 GHz CPU and 16 GB RAM. Our code is available at “https://github.com/jcchen2017/Simultaneous-Blind-Super-Resolution-and-Demixing”.

In the first experiments, we compare the proposed method with convex recovery techniques, specifically atomic norm minimization \cite{vargas2023dual,daei2024timely} and nuclear norm minimization \cite{wang2024simultaneous}. For each $k$, we generate the data matrix $\mX_k$ using $\mX_{k,\natural} = \sum_{p= 1}^{r} d_{k,p} \vh_{k,p}\va_{\tau_{k,p}}^\tran $. The locations $\{\tau_{k,p}\}_{p=1}^r$  of the $k$-th point source signal are uniformly sampled from the interval $[0,1)$. The amplitudes $\{d_{k,p}\}_{p=1}^r$ are generated as $(1+10^{c_{k,p}})e^{-\imath\varphi_{k,p}}$, where $c_{k,p}$ is uniformly sampled from $[0,1]$ and $\varphi_{k,p}$ from $[0,2\pi)$. The coefficient vectors $\vh_{k,p}$ are i.i.d. standard Gaussian random vectors, subsequently normalized. The subspace matrices $\{\mB_k\}_{k = 1}^K$ are i.i.d. random matrices with entries uniformly distributed over $[-\sqrt{3},  \sqrt{3}]$. We conduct $20$ Monte Carlo trials, deeming a recovery successful if the relative error satisfies:
\begin{align*}
	\sqrt{ \frac{\sum_{k = 1}^K\fronorm{\hat{\mX}_k - \mX_{k, \natural}}^2}{\sum_{k = 1}^K\fronorm{\mX_{k, \natural}}^2} }\leq 10^{-3}.
\end{align*}
In our tests, we fix $n=48$ and $K=2$, while varying the values of $s$ and $r$. The phase transitions of VHL, ANM, and Scaled-PG for simultaneous blind super-resolution and demixing are shown in Fig. \ref{fig:phasetransition} (a), (b), and (c), where the point source signal locations are randomly sampled. Figures \ref{fig:phasetransition} (d), (e), and (f) display the phase transition of VHL, ANM, and Scaled-PG under the separation condition $\Delta:= \min_{k\neq j}\lab\tau_k-\tau_j\rab\geq\frac{1}{n}$. In these figures, white regions represent successful recoveries, while black regions indicate failures. It is evident that Scaled-GD exhibits a slightly higher phase transition threshold compared to the VHL method, regardless of whether the separation condition is met. Additionally, the Scaled-GD approach shows less sensitivity to the separation condition compared to the ANM method.

\begin{figure}[ht!]
	\centering
	\subfigure[]{
		\includegraphics[width=0.31\linewidth]{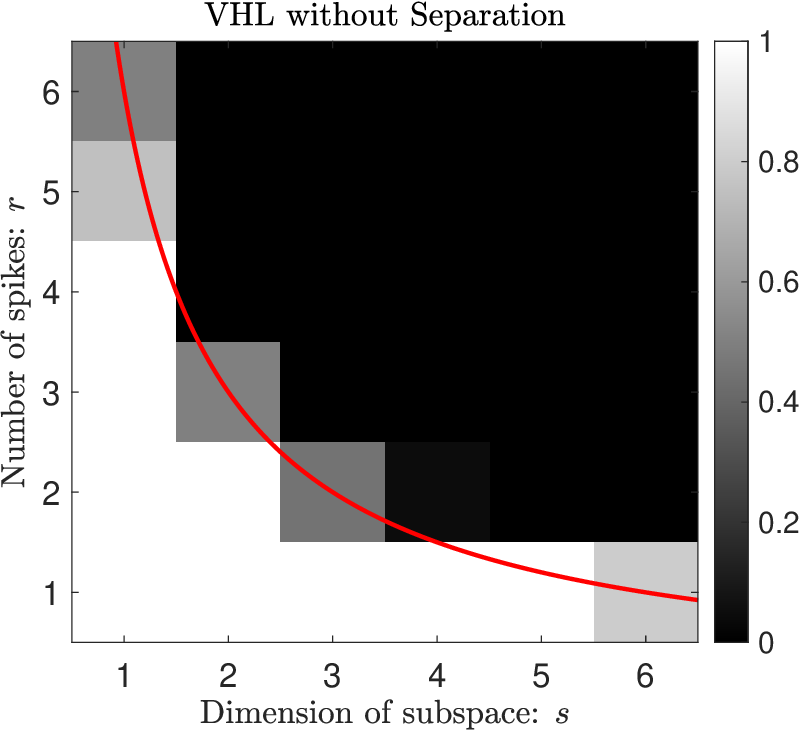}
	}
	\subfigure[]{\includegraphics[width=0.31\linewidth]{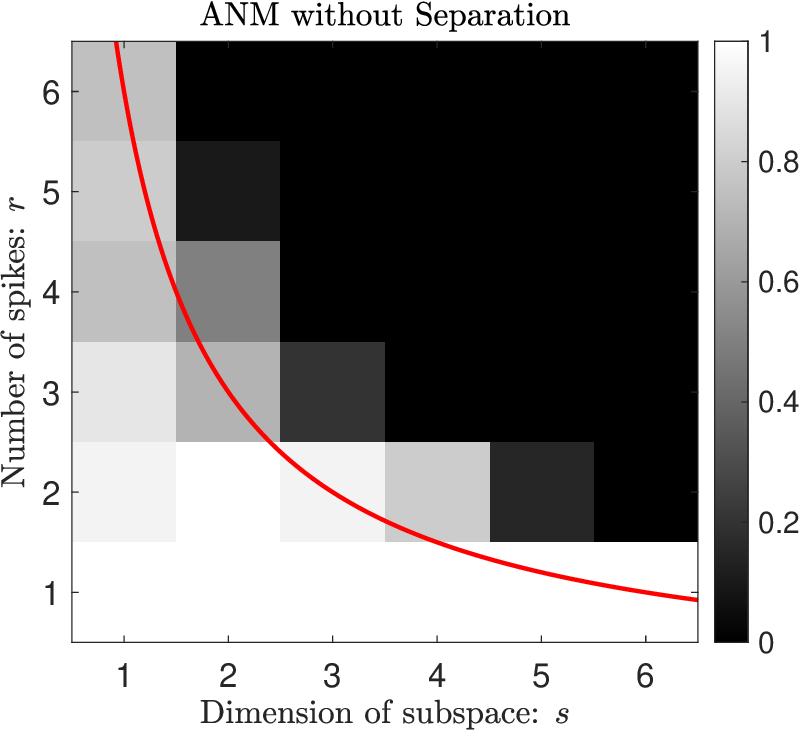}}
	\subfigure[]{\includegraphics[width=0.31\linewidth]{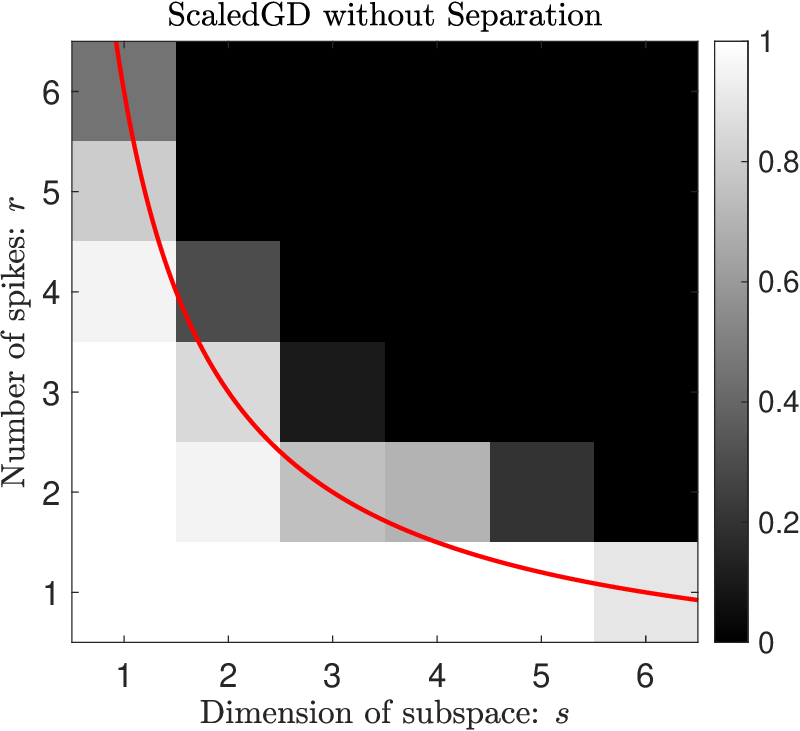}}

	\subfigure[]{
		\includegraphics[width= .31\textwidth]{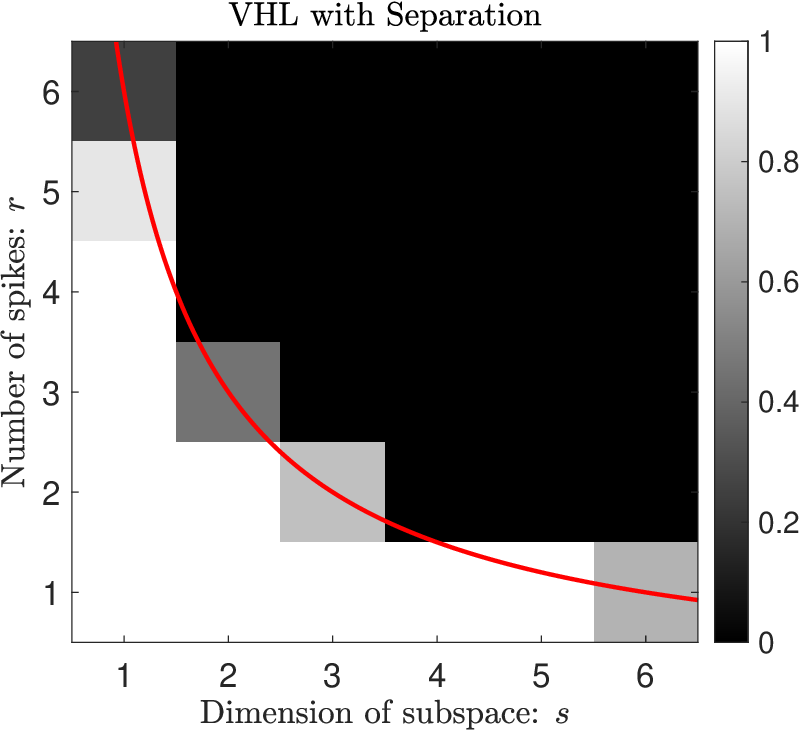}
	}
	\subfigure[]{
		\includegraphics[width= .31\textwidth]{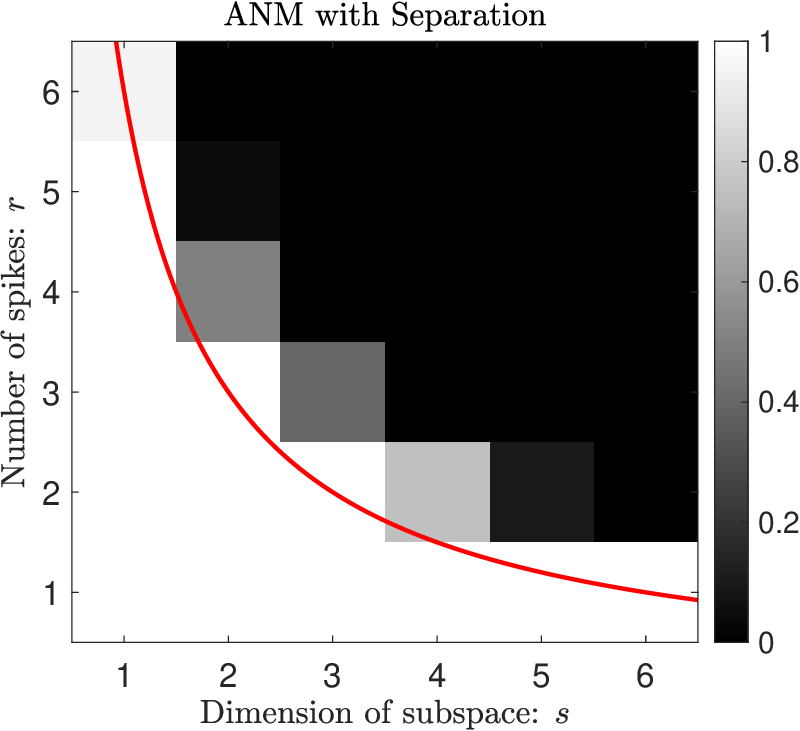}
	}
	\subfigure[]{
		\includegraphics[width= .31\textwidth]{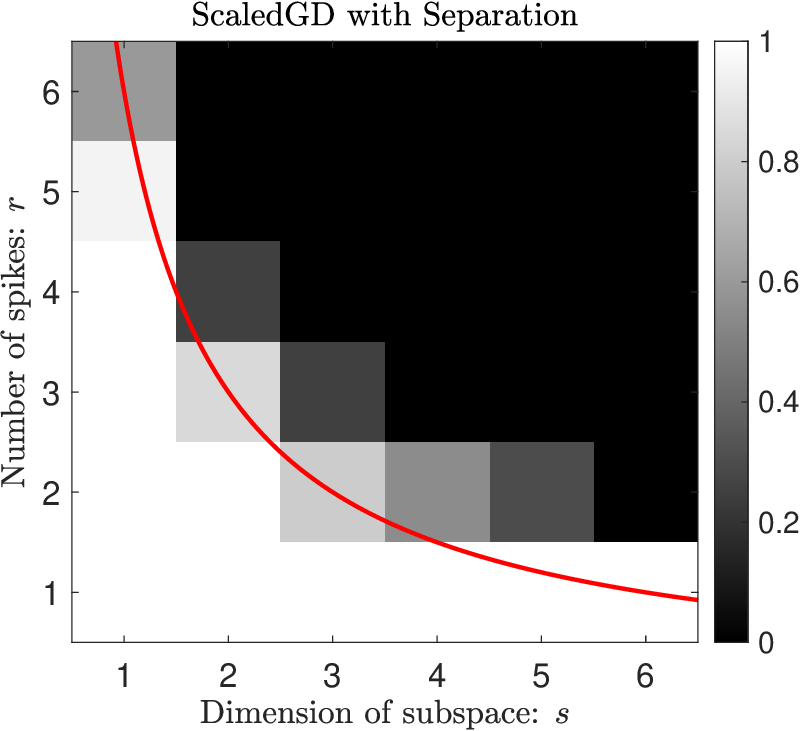}
		
	}
	
	\caption{The phase transitions of VHL, ANM and Scaled-GD were examined under the conditions  $n=48$ and $K=2$. Top: the point source signal locations are randomly generated. Bottom: the locations obey the separation condition $\Delta:= \min_{k\neq j}\lab\tau_k-\tau_j\rab\geq\frac{1}{n}$. The red curve in these figures represents the hyperbola curve $rs=6$.}
	\label{fig:phasetransition}
\end{figure}

In the second experiment, we investigate the phase transition of Scaled-GD when $s$ and $r$ are fixed. Specifically, we set $s = r = 3$ and impose the separation condition. As shown in Fig. \ref{fig:phasetransitionscaledgdwithsepnvsk}, there is an approximately linear relationship between $n$ and $K$. However, Theorem \ref{theorem: iteration} suggests that the number of measurements $n$ depends quadratically on$K$. The gap between empirical results and theoretical predictions remains an open question for future investigation.

\begin{figure}[ht!]
	\centering
	\includegraphics[width=0.5\linewidth]{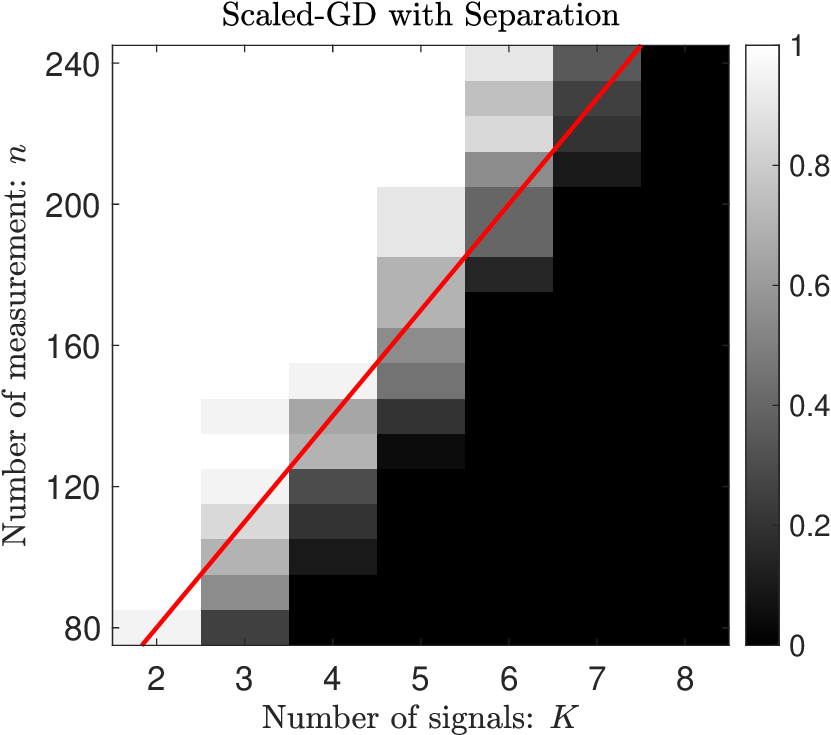}
	\caption{The phase transition of Scaled-GD for varying $n$ and $K$ when $s=r=3$. The red line in the figure represents the straight line $n=cK$, where $c$ is a constant.}
	\label{fig:phasetransitionscaledgdwithsepnvsk}
\end{figure}

In the third experiment, we examine the convergence performance of Scaled-GD in comparison with vanilla gradient descent methods. The experimental setup is similar to the first experiment, except that the data matrices are generated with a specified condition number. Specifically, the $k$-th data matrix is generated by $\mX_{k,\natural} =\sum_{p=1}^{r}d_{k,p}\vh_{k,p}\va_{\tau_{k,p}}^\tran\in\C^{s\times n}$, where $d_{k,p}=\sigma_{k,p}e^{-\imath\varphi_{k,p}}$ with $\varphi_{k,p}$ uniformly sampled from $[0,2\pi)$, and $\{\sigma_{k,p}\}_{p=1,\cdots, r}$ are linearly distributed from $1$ to $1/\kappa$. The locations $\{\tau_{k,p}\}_{p=1,\cdots, r}$ are randomly chose from $1/n, 2/n, \cdots, 1$. By applying the Vandermonde decomposition, it is evident that $\calH(\mX_{k,\natural})$ is a rank-$r$ matrix with a condition number $\kappa_k = \frac{\sigma_{k,1}}{\sigma_{k,r}}$.  For the comparison, we select $n=128, s=r=2, K=2$ and $n=512, s=r=4, K=6$. Fig \ref{fig:linearRate} illustrates the relative recovery error with respect to iterations under different condition numbers $\kappa = 1,5,10,20$. The results indicate that Scaled-GD achieves linear convergence regardless of the condition number, consistent with the predictions of our main theorem. 

\begin{figure}[ht!]
	\centering
	\subfigure[]{\includegraphics[width=0.4\linewidth]{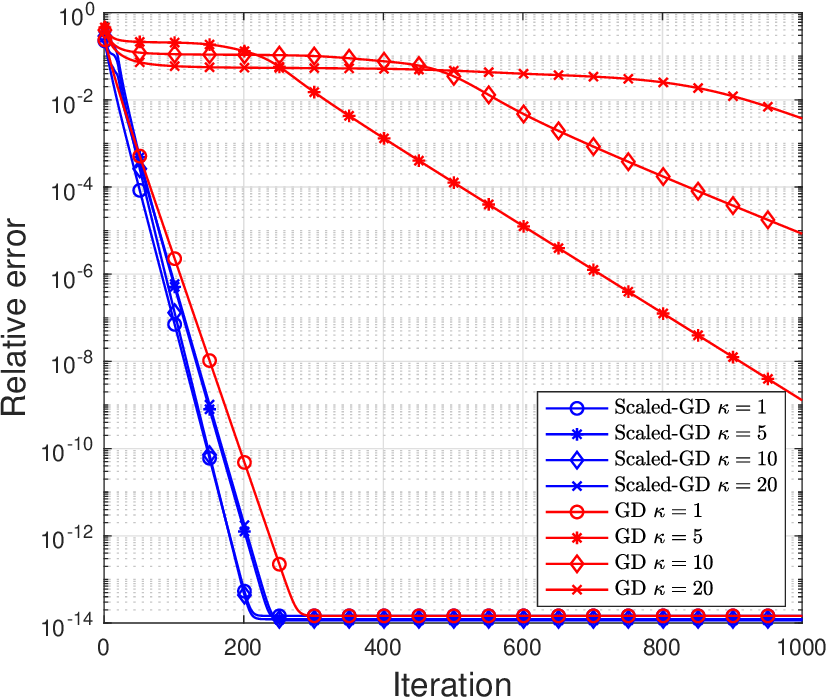} }
	\subfigure[]{\includegraphics[width=0.4\linewidth]{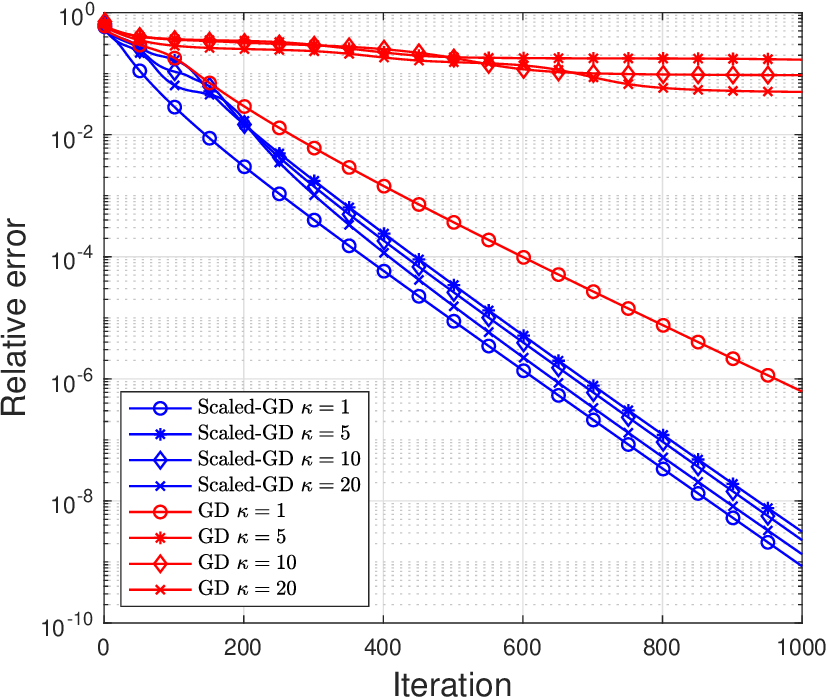}}
	\caption{The relative recovery errors of Scaled-GD and GD with respect to the iterations associated with different condition number $\kappa=1,5,15,20$. Left: $n=128, s=r=2, K=4$. Right: $n=512, s=r=4, K=6$.}
	\label{fig:linearRate}
\end{figure}

In the fourth experiment, we investigate the robustness of Scaled-GD to additive noise. The measurements are given by
\begin{align*}
	\vy' = \vy +  \sigma\cdot \frac{\vz}{\twonorm{\vz}},
\end{align*}
where $\vy$ is collected via \eqref{eq measurement}, $\sigma$ represents the noise level and $\vz$ is random vector sampled from standard norm distribution. In the test, the signal to noise ratio (SNR), defined as SNR$=20\log_{10} \frac{\twonorm{\vy}}{\sigma}$, is linearly distributed from $0$ to $60$dB. For each SNR, we choose $n=256, r=s=2, K=4$ and $n=512, r=s=4, K=6$ for comparison, with condition numbers $k=1,5, 15$. Scaled-GD is terminated after 80 iterations. The results, as depicted in Fig. \ref{fig:denoise}, show that the relative recovery error is linearly related to the noise level, demonstrating the robustness of Scaled-GD in the presence of noise.

\begin{figure}[ht!]
	\centering
	\includegraphics[width=0.5\linewidth]{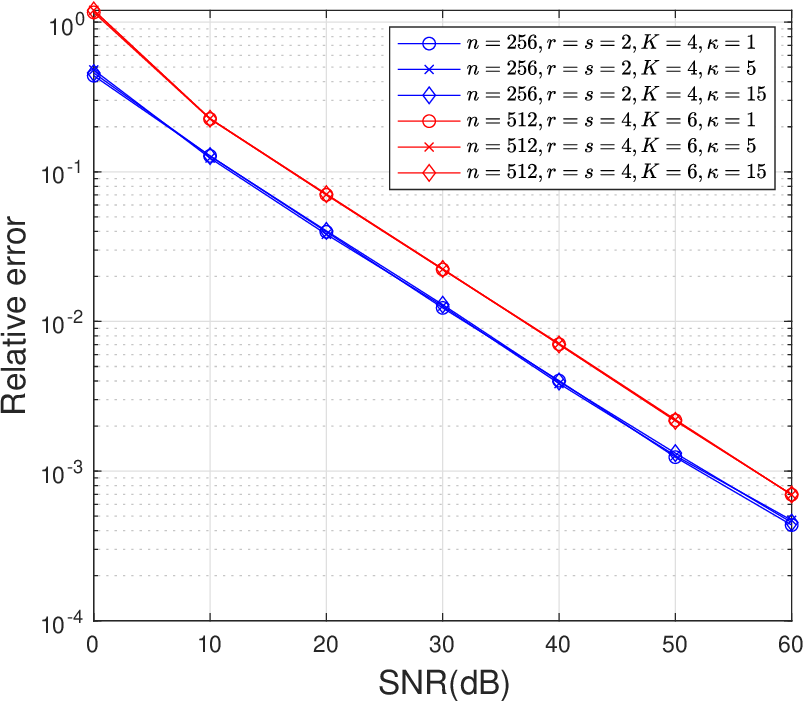}
	\caption{Performance of Scaled-GD under different SNR.}
	\label{fig:denoise}
\end{figure}

Finally, we evaluate the performance of the proposed method in recovering locations and point spread functions when combined with smoothed MUSIC. 
The data matrices are first recovered using Scaled-GD, followed by extracting the locations with smoothed MUSIC. The coefficients of point spread functions are estimated by solving the over-determined linear system. In this experiment, we set $n=256, s=r=4, K=2$. The algorithm iterates 100 times, achieving a relative recovery error of the data matrices less than $10^{-4}$. Fig. \ref{fig:figFreqs} illustrates the recovered locations and the magnitudes of the point spread functions. The results demonstrate that Scaled-GD, in conjunction with smoothed MUSIC, can accurately recover both the locations and point spread functions.

\begin{figure}[ht!]
	\centering
	\subfigure[]{\includegraphics[width=0.32\linewidth]{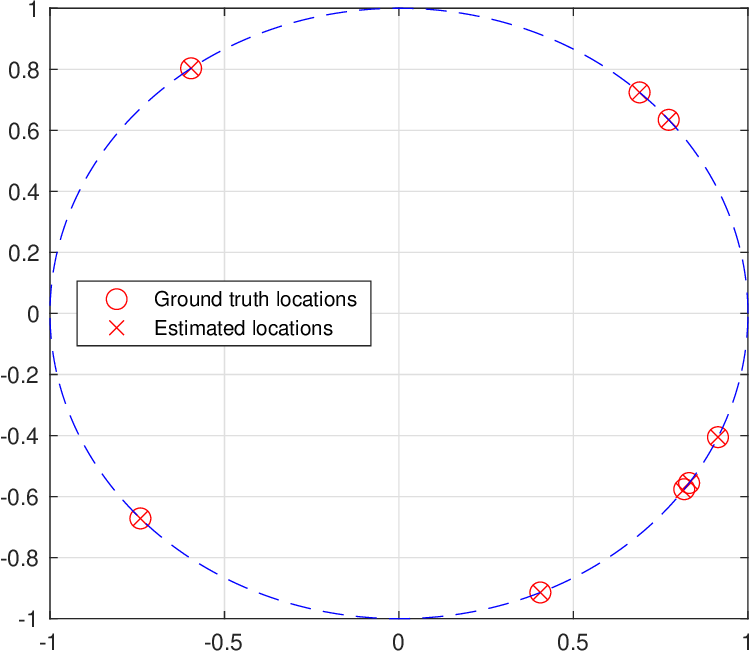}}
	\subfigure[]{
		\includegraphics[width=0.32\linewidth]{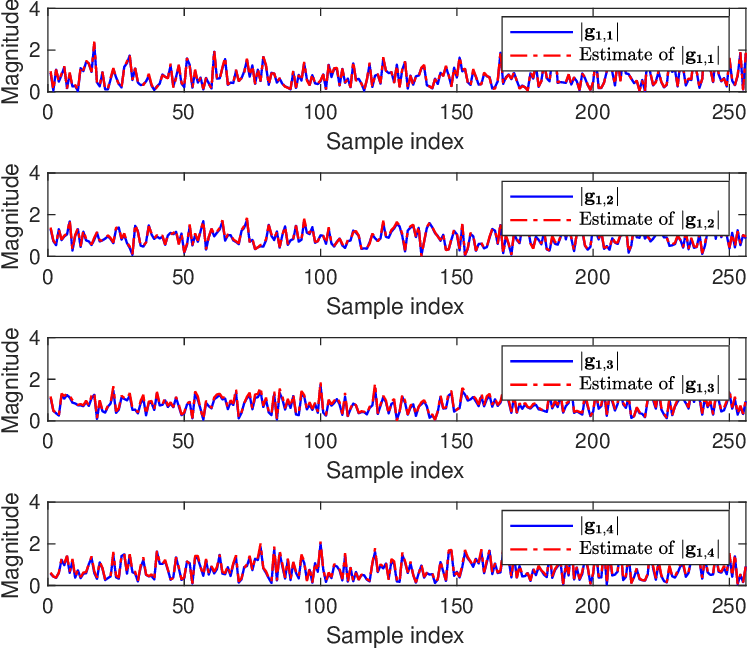}
		\includegraphics[width=0.32\linewidth]{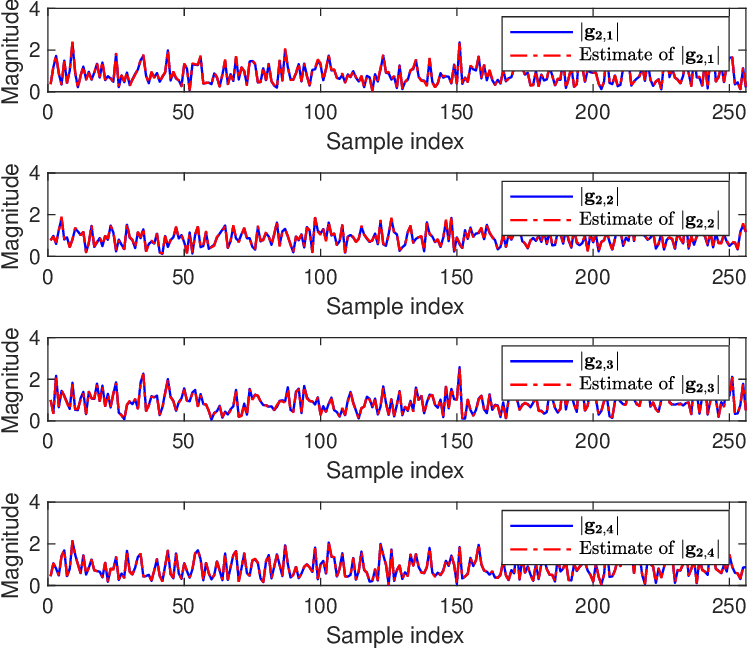} }
	\caption{Source localization and point spread function recovery with Scaled-GD and smoothed MUSIC. (a): Locations of point source signals and their estimates via smoothed MUSIC when $n=256, s=r=2, K=4$. (b) The magnitudes of $\{\vg_{k,p}\}_{k=1,2; p=1,\cdots, 4}$ and their estimates obtained through least squares.}
	\label{fig:figFreqs}
\end{figure}

\section{Proof of Main Results}
\label{sec proof main results}
To establish our main results, we first reformulate \eqref{optimization} as a block diagonal low rank matrix recovery problem.   
Let $\calA$ be a operator which maps a $Ks\times Kn$ block diagonal matrix into a $n\times 1$ vector, 
\begin{align*}
	\calA(\mX) := \sum_{k=1}^{K}\calA_k(\mX_k),
\end{align*}
where $\calA_k$ is defined in \eqref{A} and $\mX = \blkdiag(\mX_1,\cdots, \mX_K)$.
The adjoint operator of $\calA$, denoted by $\calA^\ast$ is given as 
\begin{align*}
	\calA^\ast(\vy) = \blkdiag(\calA_1^\ast(\vy), \cdots \calA^\ast_K(\vy))\in\C^{Ks\times Kn}
\end{align*}
for any vector $\vy\in\C^n$. Define $\tG:\C^{Ks\times Kn}\rightarrow \C^{Ksn_1\times Kn_2}$ as
\begin{align*}
	\tG(\mX) = \blkdiag(\calG(\mX_1), \cdots \calG(\mX_K))\in\C^{Ksn_1\times Kn_2}.
\end{align*}
The adjoint operator $\tG^\ast$  is naturally given as
\begin{align*}
	\tG^\ast(\mZ) = \blkdiag(\calGT(\mZ_1), \cdots, \calGT(\mZ_K)).
\end{align*}
It can be seen that $\tG^\ast \tG =\calI$, where $\calI:\C^{Ks\times Kn}\rightarrow \C^{Ks\times Kn}$ is the identity operator. Throughout this work, we abuse the operator 
$\calI$ to denote the identity operator in $\C^{Ks\times Kn}$ or $\C^{s\times n}$, whenever the context is clear.
Using this notation, the optimization problem \eqref{optimization} can be reformulated in a more compact form:
\begin{align}
	\label{optimization 2}
	\min_{\mL, \mR} \left\{f(\mL, \mR):=\frac{1}{2}\twonorm{\calA\tG^\ast(\mL\mR^\tranH - \mZ_{\natural})}^2 + \frac{1}{2}\fronorm{ (\calI - \tG\tG^\ast )(\mL\mR^\tranH - \mZ_{\natural}) }^2\right\},
\end{align}
where $\mZ_{\natural} =\blkdiag(\mZ_{1,\natural}, \cdots, \mZ_{K,\natural})$ is the ground truth matrix obeying that $\mD\vy = \calA\tG^\ast(\mZ_{\natural})$ and $(\calI - \tG\tG^\ast)(\mZ_{\natural})=\bzero$. Furthermore, each iteration in Algorithm \ref{alg: SGD} can be rewritten as
\begin{align*}
	\mL_{t+1} &= \mL_t - \eta_t \nabla f(\mL_t)(\mR_t^\tranH\mR_t)^{-1},\\
	\mR_{t+1} &= \mR_t - \eta_t \nabla f(\mR_t)(\mL_t^\tranH\mL_t)^{-1},
\end{align*}
where $\nabla f(\mL_t)\in\C^{Ksn_1\times Kr} $ and $\nabla f(\mR_t)\in\C^{Kn_2\times Kr}$ are given by
\begin{align}
	\label{gradient f L}
	\nabla f(\mL) &= (\mL\mR^\tranH - \mZ_{\natural})\mR +  \lb \tG\lb \calA^\ast\calA - \calI\rb\tG^\ast(\mL\mR^\tranH - \mZ_{\natural})\rb\mR,\\
	\label{gradient f R}
	\nabla f(\mR) &=(\mL\mR^\tranH - \mZ_{\natural})^\tranH\mL + \lb \tG\lb \calA^\ast\calA - \calI\rb\tG^\ast(\mL\mR^\tranH - \mZ_{\natural})^\tranH\rb\mL.
\end{align}

To establish the performance analysis for the non-convex optimization problem \eqref{optimization 2} with finite samples under the statistical model (Assumption \eqref{assumption 1}), it is instructive to first consider the more realistic case with infinite samples \cite{chen2019gradient,wang2023estimation}. Note that the expectation of the function $f(\mL,\mR)$ over the sampled data is given by:
\begin{align*}
	g(\mL, \mR):=\E[\calA]{f(\mL,\mR)} = \frac{1}{2}\twonorm{\mL\mR^\tranH - \mZ_{\natural}}^2.
\end{align*}
Therefore, the objective function in \eqref{optimization 2} can be decomposed as follows:
\begin{align*}
	f(\mL,\mR) =\underbrace{ \frac{1}{2}\fronorm{\mL\mR^\tranH - \mZ_{\natural}}^2}_{:=g(\mL,\mR)} + \underbrace{\frac{1}{2} \twonorm{\calA\tG^\ast(\mL\mR^\tranH - \mZ_{\natural})}^2 -\frac{1}{2} \fronorm{\tG\tG^\ast \lb \mL\mR^\tranH - \mZ_{\natural}\rb }^2}_{:=h(\mL,\mR)},
\end{align*}
where $h(\mL,\mR)$ can be regarded as the discrepancy between finite-sample case and the population case. In the following section, we will demonstrate that the population case $g(\mL,\mR)$ is Lipschitz smooth and satisfies the Polyak-Lojasiewicz inequality with respect to a scaled norm within a localized region. Furthermore, we will analyze the sample size required to effectively control the perturbation associated with  $h(\mL,\mR)$.

\begin{lemma}[Lipschitz smoothness of $g$]
\label{smooth of g}
Let $\mL,\mV\in\C^{Ksn_1\times r}, \mR,\mW\in\C^{Kn_2\times r}$ are block diagonal matrices, and 
\begin{align}
	\label{smooth parameter}
L=3+ \frac{2\fronorm{\mL\mR^\tranH - \mZ_{\natural} }}{\min_k \sigma_r(\mZ_{k, \natural})}  +\frac{3d_{\calS}^2(\mV, \mW)}{\min_k \sigma_r^2(\mZ_{k, \natural})} .	
\end{align}
Suppose that $\fronorm{\mL_k\mR_k^\tranH - \mZ_{k, \natural}} \leq \frac{\delta}{\sqrt{K}}\sigma_r(\mZ_{k, \natural})$ for all $k\in [K]$ with $\delta\leq \frac{1}{2}$. Then one has
\begin{align*}
	g(\mL+\mV, \mR+\mW) \leq g(\mL,\mR) + \la \nabla g(\mL), \mV\ra + \la\nabla g(\mR), \mW\ra + \frac{L}{2}d_{\calS}^2(\mV,\mW). 
\end{align*}

\end{lemma}

\begin{lemma}[Polyak-Lojasiewicz like inequality of $g$]
	\label{bound of nabla g}
	For any block diagonal matrices $\mL,\mR$, denote $\nabla g(\mL) := (\mL\mR^\tranH - \mZ_{\natural})\mR $ and $\nabla g(\mR) := (\mL\mR^\tranH - \mZ_{\natural})^\tranH\mL$. 
	Suppose that $\fronorm{\mL_k\mR_k^\tranH - \mZ_{k, \natural}} \leq \frac{\delta}{\sqrt{K}}\sigma_r(\mZ_{k, \natural})$ for any $k\in [K]$, where $\delta$ is an absolute constant such that $\delta\leq \frac{1}{(16+68\cdot 9)\kappa}$. Then one has
	\begin{align*}
		\frac{1}{3}\fronorm{\mL\mR^\tranH - \mZ_{\natural}} \leq d_{\calS^\ast}(\nabla g(\mL), \nabla g(\mR)) \leq \sqrt{2}\fronorm{\mL\mR^\tranH - \mZ_{\natural}}.
	\end{align*}
\end{lemma}

\begin{lemma}
	\label{upper bound of h}
	Suppose that $\epsilon\leq \frac{1}{15} $ and $\fronorm{\mL_k\mR_k^\tranH - \mZ_{k, \natural}}\leq \frac{\delta}{\sqrt{K}}\sigma_r(\mZ_{k, \natural})$ for all $k\in[K]$, where $\delta\leq \frac{1}{290\sqrt{\mu_0s}\kappa}$. Let 
	\begin{align*}
		\nabla h(\mL)= \lb \tG\lb \calA^\ast\calA - \calI\rb\tG^\ast(\mL\mR^\tranH - \mZ_{\natural})\rb\mR, ~
		\nabla h(\mR)= \lb \tG\lb \calA^\ast\calA - \calI\rb\tG^\ast(\mL\mR^\tranH - \mZ_{\natural})^\tranH\rb\mL.
	\end{align*}
	If $n\geq C_{\gamma}K^2s^2r^2\kappa^2\mu_0\mu_1\log^2(sn)$, then with probability at least $1-(sn)^{-\gamma}$, one has
	\begin{align*}
		d_{\calS^\ast}(\nabla h(\mL), \nabla h(\mR)) \leq \frac{1}{5}\fronorm{\mL\mR^\tranH - \mZ_{\natural}}.
	\end{align*}
\end{lemma}
\begin{lemma}
	\label{lemma bound L}
	Let $\epsilon\leq \frac{1}{6\sqrt{2}}, \delta\leq\frac{1}{116\kappa\sqrt{2\mu_0s}}$ and $ \eta\leq 1$. Suppose that $\fronorm{\mL_k\mR_k^\tranH - \mZ_{k, \natural}} \leq \frac{\delta}{\sqrt{K}}\sigma_r(\mZ_{k, \natural})$ for any $k\in[K]$. 
	If the block diagonal matrices $\mV$ and $\mW$ are given by $\mV = -\eta \nabla f(\mL)(\mR^\tranH\mR)^{-1}$ and $\mW=-\eta \nabla f(\mR)(\mL^\tranH\mL)^{-1}$, then the smoothness parameter $L$ defined in \eqref{smooth parameter} obeys that $3\leq L\leq 5$.
\end{lemma}

\begin{lemma}[Contraction of $g$]
	\label{lemma linear convergence}
	Suppose that $\fronorm{\mL_k \mR_k^\tranH - \mZ_{k,\natural}} \leq \frac{\delta}{\sqrt{K}}\sigma_r(\mZ_{k,\natural})$ for any $k\in[K]$. 
	Let $0< \eta_t \leq \frac{1}{20}$. Then the objective function obeys that
	\begin{align*}
		g(\mL_{t+1}, \mR_{t+1} )\leq \lb 1- \frac{\eta_t}{90}\rb g(\mL_t, \mR_t).
	\end{align*}
\end{lemma}

\begin{lemma}[Initialization]
	\label{theorem: initial}
	If the number of measurements satisfies $n\geq C_{\gamma}\rho^{-2} K^2s^2 r^2\mu_0 \mu_1  \kappa^2\log^2(sn)$,
	then,  with probability at least $1-(sn)^{-\gamma+1}$, there holds the following inequalities
	\begin{align*}
		\fronorm{\mZ_{\ell, 0} - \mZ_{\ell, \natural}} \leq\frac{\rho}{\sqrt{Ks}}\sigma_r(\mZ_{\ell, \natural}) ~ \forall  1\leq \ell\leq K, 
	\end{align*}
	where $\{\mZ_{\ell,0}\}_{\ell=1,\cdots, K}$ are obtained via Algorithm \ref{alg: SGD}.
	
\end{lemma}
The proofs of these lemmas are deferred to Section \ref{proof key lemmas}.

\subsection{Proof of Theorem \ref{theorem: iteration}}
\label{proof iteration}
Now, we are in the position to prove our main results by induction. Firstly, we show \eqref{variable error} holds when $t=0$. According to Lemma \ref{theorem: initial}, one has     
\begin{align*}
\fronorm{\mL_{k,0}\mR_{k,0}^\tranH - \mZ_{k,\natural} }^2 \leq \frac{\rho^2}{K} \sigma_r^2(\mZ_{k,\natural}),\quad \forall 1\leq k\leq K.
\end{align*}
provided that $n\geq C_{\gamma}\rho^2K^2s^2r^2\kappa^2\mu_0\mu_1\log^2(sn)$. Thus we have
\begin{align*}
	\sum_{k=1}^{K} \fronorm{\mL_{k,0}\mR_{k,0}^\tranH - \mZ_{k,\natural} }^2 \leq \frac{\rho^2}{K} \sum_{k=1}^{K} \sigma_r^2(\mZ_{k,\natural}) = \frac{\rho^2\sigma^2_0}{K}.
\end{align*}
Moreover, we assume \eqref{variable error} holds for the iterations $0,1,\cdots, t$, and then prove it also hold for $t+1$. It is worth noting that 
\begin{align*}
	\lb 1- \frac{\eta_t}{90}\rb^t\frac{\delta^2\sigma_0^2}{K}&\geq\sum_{k=1}^{K}\fronorm{\mL_{k,t}\mR_{k,t}^\tranH-\mZ_{k, \natural}}^2 =\sum_{k=1}^{K}\sigma_r^2(\mZ_{k,\natural})\frac{\fronorm{\mL_{k,t}\mR_{k,t}^\tranH-\mZ_{k, \natural}}^2 }{\sigma_r^2(\mZ_{k,\natural})} \\
	& \geq \max_k \frac{\fronorm{\mL_{k,t}\mR_{k,t}^\tranH-\mZ_{k, \natural}}^2 }{\sigma_r^2(\mZ_{k,\natural})} \cdot \sigma^2_0,
\end{align*}
which implies that 
\begin{align*}
	\max_k \frac{\fronorm{\mL_{k,t}\mR_{k,t}^\tranH-\mZ_{k, \natural}}^2 }{\sigma_r^2(\mZ_{k,\natural})}  \leq \lb 1- \frac{\eta_t}{90}\rb^t\frac{\delta^2}{K}\leq  \frac{\delta^2}{K}.
\end{align*}
Therefore, Lemma \ref{lemma linear convergence} implies that 
\begin{align*}
	\sum_{k=1}^{K} \fronorm{\mL_{k,t+1}\mR_{k,t+1}^\tranH - \mZ_{k,\natural} }^2 \leq \lb 1-\frac{\eta_t}{90}\rb \sum_{k=1}^{K} \fronorm{\mL_{k,t}\mR_{k,t}^\tranH - \mZ_{k,\natural} }^2 \leq \lb 1-\frac{\eta_t}{90}\rb^{t+1} \frac{\delta^2\sigma_0^2}{K},
\end{align*}
which completes the proof.

\section{Conclusion}
\label{sec conclusion}
In this work, we investigate the simultaneous blind super-resolution and demixing problem, formulating it as a low-rank matrix recovery problem. We introduce a novel non-convex algorithm, termed scaled gradient descent without balance, specifically designed to address this problem. We establish the sample complexity and provide a linear convergence guarantee for the proposed algorithm, notably demonstrating that the convergence rate is independent of the condition number of the target matrices. The empirical performance of our algorithm has been validated through extensive experiments.

\subsection*{Acknowledge}
The author expresses gratitude to Jinsheng Li for identifying an error during the preparation of this work, which has significantly improved its quality. Additionally, the author thanks Ke Wei and Xu Zhang for their valuable discussions.

\appendix
\section{Proof of key lemmas}
\label{proof key lemmas}
\subsection{Proof of Lemma \ref{smooth of g}}
A direct computation yields that 
\begin{align*}
	(\mL+\mV)(\mR+\mW)^\tranH - \mZ_{\natural} = \mL\mR^\tranH - \mZ_{\natural} + \mL\mW^\tranH + \mV\mR^\tranH + \mV\mW^\tranH.
\end{align*}
Then the objective function $g$ at $(\mL+\mV, \mR+\mW)$ can be expressed as follows:
\begin{align*}
	g(\mL+\mV, \mR+\mW) &= \frac{1}{2} \fronorm{(\mL+\mV)(\mR+\mW)^\tranH - \mZ_{\natural}}^2 \\
	&=\frac{1}{2} \fronorm{\mL\mR^\tranH - \mZ_{\natural} + \mL\mW^\tranH + \mV\mR^\tranH + \mV\mW^\tranH}^2\\
	&=\frac{1}{2} \fronorm{\mL\mR^\tranH - \mZ_{\natural} }^2+ \la\mL\mR^\tranH - \mZ_{\natural} ,  \mL\mW^\tranH + \mV\mR^\tranH + \mV\mW^\tranH\ra +\frac{1}{2} \fronorm{ \mL\mW^\tranH + \mV\mR^\tranH + \mV\mW^\tranH}^2\\
	&=g(\mL,\mR)+ \la \nabla g(\mL), \mV\ra + \la \nabla g(\mR), \mW\ra\\
	&\quad +\la\mL\mR^\tranH - \mZ_{\natural} ,  \mV\mW^\tranH\ra +\frac{1}{2} \fronorm{ \mL\mW^\tranH + \mV\mR^\tranH + \mV\mW^\tranH}^2\\
	&\leq g(\mL,\mR)+ \la \nabla g(\mL), \mV\ra + \la \nabla g(\mR), \mW\ra\\
	&\quad +\underbrace{\fronorm{\mL\mR^\tranH - \mZ_{\natural} }\cdot \fronorm{ \mV\mW^\tranH} +\frac{3}{2} \lb \fronorm{ \mL\mW^\tranH }^2+ \fronorm{ \mV\mR^\tranH}^2  +\fronorm{ \mV\mW^\tranH}^2\rb}_{:=\Gamma},
\end{align*}
where we have used the fact that $\nabla g(\mL) = (\mL\mR^\tranH - \mZ_{\natural})\mR$ and $\nabla g(\mL) = (\mL\mR^\tranH - \mZ_{\natural})^\tranH \mL$.  In what follows, we control $\Gamma$.
\begin{itemize}
	\item Bounding of $\fronorm{ \mL\mW^\tranH }^2$ and  $\fronorm{ \mV\mR^\tranH}^2$. 
	Since $\mL = \blkdiag(\mL_1, \cdots, \mL_K)$ and $\mW =\blkdiag(\mW_1,\cdots, \mW_K)$, we have
	\begin{align*}
		\fronorm{\mL\mW^\tranH}^2 &=\sum_{k=1}^{K} \fronorm{\mL_k\mW_k^\tranH}^2 =\sum_{k=1}^{K} \fronorm{\mL_k(\mL_k^\tranH\mL_k)^{-\frac{1}{2}} (\mL_k^\tranH\mL_k)^{\frac{1}{2}}\mW_k^\tranH}^2\\
		&\leq \sum_{k=1}^{K} \lb \opnorm{\mL_k(\mL_k^\tranH\mL_k)^{-\frac{1}{2}}}\cdot \fronorm{ \mW_k(\mL_k^\tranH\mL_k)^{\frac{1}{2}}}^2 \rb\\
		&\leq \sum_{k=1}^{K}\fronorm{\mW_k(\mL_k^\tranH\mL_k)^{\frac{1}{2}}}^2\\
		&=\fronorm{\mW(\mL^\tranH\mL)^{\frac{1}{2}}}^2, \numberthis\label{LW}
	\end{align*}
	where we have used the fact $  \opnorm{\mL_k(\mL_k^\tranH\mL_k)^{-\frac{1}{2}}}  \leq 1$. Similarly, one has
	\begin{align*}		
		\fronorm{\mV\mR^\tranH}^2&\leq \fronorm{\mV(\mR^\tranH\mR)^{\frac{1}{2}}}^2.\numberthis\label{VR}
	\end{align*}
	Combining \eqref{LW} and \eqref{VR} together yields that 
	\begin{align*}
		\fronorm{\mL\mW^\tranH}^2 + \fronorm{\mV\mR^\tranH}^2 \leq \fronorm{\mW(\mL^\tranH\mL)^{\frac{1}{2}}}^2 + \fronorm{\mV(\mR^\tranH\mR)^{\frac{1}{2}}}^2 = d_{\calS}^2(\mV,\mW).\numberthis\label{LWVR}
	\end{align*}
	\item Bounding of $\fronorm{ \mV\mW^\tranH}$. A direct computation yields that
	\begin{align*}
		\fronorm{\mV\mW^\tranH} &= 	\fronorm{\mV(\mR^\tranH\mR)^{\frac{1}{2}}(\mR^\tranH\mR)^{-\frac{1}{2}}(\mL^\tranH\mL)^{-\frac{1}{2}}(\mL^\tranH\mL)^{\frac{1}{2}}\mW^\tranH}\\
		&\leq \fronorm{\mV(\mR^\tranH\mR)^{\frac{1}{2}}}\cdot \opnorm{(\mR^\tranH\mR)^{-\frac{1}{2}}(\mL^\tranH\mL)^{-\frac{1}{2}}}\cdot \fronorm{\mW(\mL^\tranH\mL)^{\frac{1}{2}}}\\
		&\stackrel{(a)}{\leq}\frac{1}{(1-\delta)\min_k \sigma_r(\mZ_{k, \natural})}\fronorm{\mV(\mR^\tranH\mR)^{\frac{1}{2}}}\cdot\fronorm{\mW(\mL^\tranH\mL)^{\frac{1}{2}}}\\
		&\stackrel{(b)}{\leq }\frac{1}{2(1-\delta)\min_k \sigma_r(\mZ_{k, \natural})}\lb \fronorm{\mV(\mR^\tranH\mR)^{\frac{1}{2}}}^2 + \fronorm{\mW(\mL^\tranH\mL)^{\frac{1}{2}}}^2 \rb\\
		&=\frac{1}{2(1-\delta)\min_k \sigma_r(\mZ_{k, \natural})} d_{\calS}^2(\mV,\mW)\\
		&\leq  \frac{1}{\min_k \sigma_r(\mZ_{k, \natural})} d_{\calS}^2(\mV,\mW),\numberthis\label{VW},
	\end{align*}
	where step (a) is due to the fact that
	\begin{align}
		\label{fact 1}
		\opnorm{(\mR^\tranH\mR)^{-\frac{1}{2}} (\mL^\tranH\mL)^{-\frac{1}{2}} }\leq \frac{1}{(1-\delta)\min_k \sigma_r(\mZ_{k, \natural})},
	\end{align}
	step (b) follows from that $ab\leq \frac{1}{2}(a^2+b^2)$ for any $a,b$, and the last line is due to $\delta\leq \frac{1}{2}$. Moreover, the fact \eqref{fact 1} can be proved as follows: Firstly, the Weyl's theorem yields that 
	\begin{align*}
		\left| \sigma_r(\mL_k\mR_k^\tranH) - \sigma_r(\mZ_{k, \natural})\right| \leq \opnorm{\mL_k\mR_k^\tranH - \mZ_{k, \natural}}\leq  \fronorm{\mL_k\mR_k^\tranH - \mZ_{k, \natural}} \leq \frac{\delta}{\sqrt{K}}\sigma_r(\mZ_{k, \natural})\leq \delta\sigma_r(\mZ_{k, \natural}),
	\end{align*}
	which implies that 
	\begin{align*}
		\sigma_r(\mL_k\mR_k^\tranH) &\geq \lb1- \delta \rb\sigma_r(\mZ_{k, \natural}),\\
		\min_k \sigma_r(\mL_k\mR_k^\tranH) &\geq (1-\delta)\min_k \sigma_r(\mZ_{k, \natural}).
	\end{align*}
	Secondly, a direct computation yields that
	\begin{align*}
		\opnorm{(\mR^\tranH\mR)^{-\frac{1}{2}} (\mL^\tranH\mL)^{-\frac{1}{2}} }&=\opnorm{\mR(\mR^\tranH\mR)^{-1}(\mL^\tranH\mL)^{-1}\mL}\\
		&=\opnorm{{\mR^\tranH}^\dagger \mL^\dagger}\\
		&=\opnorm{\begin{bmatrix}
				(\mL_1\mR_1^\tranH)^\dagger &&\\
				&\ddots &\\
				&&(\mL_K\mR_K^\tranH)^\dagger\\
		\end{bmatrix}}\\
		&=\frac{1}{\min_k \sigma_r(\mL_k\mR_k^\tranH)} \\
		&\leq \frac{1}{(1-\delta)\min_k \sigma_r(\mZ_{k, \natural})},
	\end{align*}
\end{itemize}
Combining together yields that
\begin{align*}
	\Gamma &=\fronorm{\mL\mR^\tranH - \mZ_{\natural} }\cdot \fronorm{ \mV\mW^\tranH} +\frac{3}{2} \lb \fronorm{ \mL\mW^\tranH }^2+ \fronorm{ \mV\mR^\tranH}^2  +\fronorm{ \mV\mW^\tranH}^2\rb\\
	&\leq \frac{\fronorm{\mL\mR^\tranH - \mZ_{\natural} }}{\min_k \sigma_r(\mZ_{k, \natural})} d_{\calS}^2(\mV,\mW) + \frac{3}{2} d_{\calS}^2(\mV,\mW) +  \frac{3}{2\min_k \sigma_r^2(\mZ_{k, \natural})} d_{\calS}^4(\mV,\mW)\\
	&=\lb \frac{3}{2} + \frac{\fronorm{\mL\mR^\tranH - \mZ_{\natural} }}{\min_k \sigma_r(\mZ_{k, \natural})}  +\frac{3d_{\calS}^2(\mV, \mW)}{2\min_k \sigma_r^2(\mZ_{k, \natural})} \rb d_{\calS}^2(\mV, \mW)\\
	:&=\frac{L}{2} d_{\calS}^2(\mV, \mW),
\end{align*}
where $L=3+ \frac{2\fronorm{\mL\mR^\tranH - \mZ_{\natural} }}{\min_k \sigma_r(\mZ_{k, \natural})}  +\frac{3d_{\calS}^2(\mV, \mW)}{\min_k \sigma_r^2(\mZ_{k, \natural})} $.
Thus we have
\begin{align*}
	g(\mL+\mV, \mR+\mW) &\leq g(\mL,\mR) + \la \nabla g(\mL), \mV\ra +\la\nabla g(\mR), \mW\ra+ \frac{L}{2} d_{\calS}^2(\mV, \mW),
\end{align*}
which completes the proof.

\subsection{Proof of Lemma \ref{bound of nabla g}}
By the definition of $d_{\calS^\ast}(\cdot, \cdot)$ in \eqref{def dual ds}, we have
	\begin{align*}
		d_{\calS^\ast}^2(\nabla g(\mL), \nabla g(\mR)) &=\fronorm{ \lb \mL\mR^\tranH - \mZ_{\natural} \rb \mR(\mR^\tranH\mR)^{-\frac{1}{2}} }^2+\fronorm{ \lb \mL\mR^\tranH - \mZ_{\natural} \rb^\tranH \mL(\mL^\tranH\mL)^{-\frac{1}{2}} }^2 \\
		&\leq 2\fronorm{\mL\mR^\tranH - \mZ_{\natural}}^2 ,
	\end{align*}
	where we have used the fact that $\opnorm{\mR(\mR^\tranH\mR)^{-\frac{1}{2}}}\leq 1$ and $\opnorm{\mL(\mL^\tranH\mL)^{-\frac{1}{2}}}\leq 1$. 
	
	Secondly, we show the lower bound. Notice that 
	\begin{align*}
		d_{\calS^\ast}(\nabla g(\mL), \nabla g(\mR))&=\sqrt{\fronorm{ \lb \mL\mR^\tranH - \mZ_{\natural} \rb \mR(\mR^\tranH\mR)^{-\frac{1}{2}} }^2+\fronorm{ \lb \mL\mR^\tranH - \mZ_{\natural} \rb^\tranH \mL(\mL^\tranH\mL)^{-\frac{1}{2}} }^2}\\
		&=\fronorm{\begin{bmatrix}
				\lb \mL\mR^\tranH - \mZ_{\natural} \rb \mR(\mR^\tranH\mR)^{-\frac{1}{2}}\\
				\lb \mL\mR^\tranH - \mZ_{\natural} \rb^\tranH \mL(\mL^\tranH\mL)^{-\frac{1}{2}} \\
		\end{bmatrix}}\\
		&=\sup_{\mY = \begin{bmatrix}
				\mY_1^\tran &\mY_2^\tran
			\end{bmatrix}^\tran:\fronorm{\mY}=1} \left| \la  \begin{bmatrix}
			\lb \mL\mR^\tranH - \mZ_{\natural} \rb \mR(\mR^\tranH\mR)^{-\frac{1}{2}}\\
			\lb \mL\mR^\tranH - \mZ_{\natural} \rb^\tranH \mL(\mL^\tranH\mL)^{-\frac{1}{2}} \\
		\end{bmatrix}, \begin{bmatrix}
			\mY_1\\
			\mY_2\\
		\end{bmatrix}\ra\right|\\
		&=\sup_{\mY = \begin{bmatrix}
				\mY_1^\tran &\mY_2^\tran
			\end{bmatrix}^\tran:\fronorm{\mY}=1} \left| \la \mL\mR^\tranH - \mZ_{\natural} , \mY_1(\mR^\tranH\mR)^{-\frac{1}{2}}\mR^\tranH+ \mL(\mL^\tranH\mL)^{-\frac{1}{2}}\mY_2\ra  \right|.
	\end{align*}
	Consider a spectral $\mY$ such that
	\begin{align*}
		\mY &= \frac{1}{d_{\calS}(\mDelta, \mUpsilon)}\begin{bmatrix}
			\mDelta(\mR^\tranH\mR)^{\frac{1}{2}}\\
			\mUpsilon(\mL^\tranH\mL)^{\frac{1}{2}}\\
		\end{bmatrix},
	\end{align*}
	where $\mDelta=\mL\mO - \mL_{\natural} $ and $\mUpsilon=\mR\mO - \mR_{\natural}$. It is easy to see that $\fronorm{\mY}=1$. As a result, we have
	\begin{align*}
		d_{\calS^\ast}(\nabla g(\mL), \nabla g(\mR))&\geq \frac{1}{d_{\calS}(\mDelta, \mUpsilon)} \left| \la \mL\mR^\tranH - \mZ_{\natural} , \mDelta\mR^\tranH+ \mL\mUpsilon^\tranH\ra  \right|.
	\end{align*}
	A direct computation yields that 
	\begin{align*}
		\mL\mR^\tranH - \mZ_{\natural} &= (\mL_{\natural}+\mDelta)(\mR_{\natural}+\mUpsilon)^\tranH - \mL_{\natural}\mR_{\natural}^\tranH = \mL_{\natural}\mUpsilon^\tranH + \mDelta\mR_{\natural}^\tranH + \mDelta\mUpsilon^\tranH ,\\
		\mDelta\mR^\tranH +\mL\mUpsilon^\tranH &=\mDelta(\mR_{\natural}+\mUpsilon)^\tranH + (\mL_{\natural}+\mDelta)\mUpsilon^\tranH = \mL_{\natural}\mUpsilon^\tranH + \mDelta\mR_{\natural}^\tranH + 2\mDelta\mUpsilon^\tranH.
	\end{align*}
	Thus
	\begin{align*}
		\left| \la \mL\mR^\tranH - \mZ_{\natural} , \mDelta\mR^\tranH+ \mL\mUpsilon^\tranH\ra  \right| &= \left| \la \mL_{\natural}\mUpsilon^\tranH + \mDelta\mR_{\natural}^\tranH + \mDelta\mUpsilon^\tranH,\mL_{\natural}\mUpsilon^\tranH + \mDelta\mR_{\natural}^\tranH + 2\mDelta\mUpsilon^\tranH \ra\right|\\
		&= \fronorm{\mL_{\natural}\mUpsilon^\tranH + \mDelta\mR_{\natural}^\tranH}^2 + 2\fronorm{ \mDelta\mUpsilon^\tranH}^2 +3\la \mL_{\natural}\mUpsilon^\tranH + \mDelta\mR_{\natural}^\tranH,  \mDelta\mUpsilon^\tranH\ra_R\\
		&\geq \fronorm{\mL_{\natural}\mUpsilon^\tranH + \mDelta\mR_{\natural}^\tranH}^2 + 2\fronorm{ \mDelta\mUpsilon^\tranH}^2 -3\fronorm{\mL_{\natural}\mUpsilon^\tranH + \mDelta\mR_{\natural}^\tranH}\cdot \fronorm{  \mDelta\mUpsilon^\tranH}\\
		&\geq \frac{3}{4}\fronorm{\mL_{\natural}\mUpsilon^\tranH + \mDelta\mR_{\natural}^\tranH}^2 -7\fronorm{ \mDelta\mUpsilon^\tranH}^2,\numberthis \label{tmp2}
	\end{align*}
	where the last line is due to $a^2 + 2b^2 -3ab\geq \lb1-\frac{3 \alpha}{2}\rb a^2 - \lb \frac{3}{2\alpha} -2\rb b^2$  with $\alpha=\frac{1}{6}$.
	According to Lemma \ref{lemma phi and psi} , one has
	\begin{align*}
		\fronorm{\mL_{\natural}\mUpsilon^\tranH + \mDelta\mR_{\natural}^\tranH} &\geq \lb 1-\frac{6\kappa\delta}{\sqrt{K}}\rb  \fronorm{\mL\mR^\tranH - \mZ_{\natural}} \geq (1-6\kappa\delta)\fronorm{\mL\mR^\tranH - \mZ_{\natural}}\geq0,\\
		\fronorm{\mDelta\mUpsilon^\tranH} &\leq \frac{6\kappa\delta}{\sqrt{K}}\fronorm{\mL\mR^\tranH - \mZ_{\natural}}\leq 6\kappa\delta\fronorm{\mL\mR^\tranH - \mZ_{\natural}}
	\end{align*}
	provided that $\delta\leq \frac{1}{6\kappa}$. Plugging these into \eqref{tmp2} yields that
	\begin{align*}
		\left| \la \mL\mR^\tranH - \mZ_{\natural} , \mDelta\mR+ \mL\mUpsilon^\tranH\ra  \right| &\geq\frac{3}{4} \fronorm{\mL_{\natural}\mUpsilon^\tranH + \mDelta\mR_{\natural}^\tranH}^2 -7 \fronorm{ \mDelta\mUpsilon^\tranH}^2\\
		&\geq \lb \frac{3}{4} (1-6\kappa\delta)^2- 7\cdot (6\kappa\delta)^2\rb \fronorm{\mL\mR^\tranH - \mZ_{\natural}}^2\\
		&\geq \lb \frac{3}{4}(1-12\kappa\delta) -  7\cdot 6\kappa\delta\rb \fronorm{\mL\mR^\tranH - \mZ_{\natural}}^2\\
		&=\lb \frac{3}{4} - 51\kappa\delta \rb\fronorm{\mL\mR^\tranH - \mZ_{\natural}}^2,
	\end{align*}
	where we have used that $\delta\leq\frac{1}{6\kappa}$, i.e, $(6\kappa\delta)^2\leq 6\kappa\delta\leq 1$.
	Moreover, one has
	\begin{align*}
		d_{\calS}^2(\mDelta, \mUpsilon) &=\fronorm{\mDelta(\mR^\tranH\mR)^{\frac{1}{2}}}^2 + \fronorm{\mUpsilon(\mL^\tranH\mL)^{\frac{1}{2}}}^2\\
		&=\fronorm{\mDelta(\mR^\tranH\mR)^{\frac{1}{2}}}^2 + \fronorm{\mUpsilon(\mL^\tranH\mL)^{\frac{1}{2}}}^2.
	\end{align*}
	Notice that 
	\begin{align*}
		\fronorm{\mDelta(\mR^\tranH\mR)^{\frac{1}{2}}}&= \fronorm{\mDelta\bSigma_{\natural}^{\frac{1}{2}}\bSigma_{\natural}^{-\frac{1}{2}}(\mR^\tranH\mR)^{\frac{1}{2}}}\\
		&\leq \fronorm{\mDelta\bSigma_{\natural}^{\frac{1}{2}}}\cdot \opnorm{\bSigma_{\natural}^{-\frac{1}{2}}(\mR^\tranH\mR)^{\frac{1}{2}}}\\
		&=\fronorm{\mDelta\bSigma_{\natural}^{\frac{1}{2}}}\cdot \opnorm{\mR\bSigma_{\natural}^{-\frac{1}{2}}}\\
		&\leq \fronorm{\mDelta\bSigma_{\natural}^{\frac{1}{2}}}\cdot \opnorm{(\mR_{\natural}+\mUpsilon)\bSigma_{\natural}^{-\frac{1}{2}}}\\
		&\leq \fronorm{\mDelta\bSigma_{\natural}^{\frac{1}{2}}}\cdot \lb 1+ \opnorm{ \mUpsilon \bSigma_{\natural}^{-\frac{1}{2}}} \rb\\
		&\stackrel{(a)}{\leq }\lb 1+\frac{2\kappa\delta}{\sqrt{K}}\rb\fronorm{\mDelta\bSigma_{\natural}^{\frac{1}{2}}},
	\end{align*}
	where step (a) follows from Lemma \ref{lemma a}. Using the same argument, one has
	\begin{align*}
		\fronorm{\mUpsilon(\mL^\tranH\mL)^{\frac{1}{2}}} &\leq \lb 1+\frac{2\kappa\delta}{\sqrt{K}}\rb\fronorm{\mUpsilon\bSigma_{\natural}^{\frac{1}{2}}}.
	\end{align*}
	Combining together yields that 
	\begin{align*}
		d_{\calS}(\mDelta, \mUpsilon) &\leq \sqrt{\fronorm{\mDelta(\mR^\tranH\mR)^{\frac{1}{2}}}^2 +\fronorm{\mUpsilon(\mL^\tranH\mL)^{\frac{1}{2}}}^2 }\\
		&\leq \lb 1+\frac{2\kappa\delta}{\sqrt{K}}\rb \sqrt{\fronorm{\mDelta\bSigma_{\natural}^{\frac{1}{2}}}^2 +  \fronorm{\mUpsilon\bSigma_{\natural}^{\frac{1}{2}}}^2} \\
		&\leq \lb 1+\frac{2\kappa\delta}{\sqrt{K}}\rb \sqrt{(\sqrt{2}+1)}\fronorm{\mL\mR^\tranH - \mZ_{\natural}}\\
		&\leq 2\lb 1+\frac{2\kappa\delta}{\sqrt{K}}\rb\fronorm{\mL\mR^\tranH - \mZ_{\natural}}\\
		&\leq 2\lb 1+2\kappa\delta \rb\fronorm{\mL\mR^\tranH - \mZ_{\natural}},
	\end{align*}
	where the third line is due to Lemma \ref{lemma 24 cong}. Therefore, one has
	\begin{align*}
		d_{\calS^\ast}(\nabla g(\mL), \nabla g(\mR)) &\geq \frac{1}{d_{\calS}(\mDelta, \mUpsilon)}\left| \la \mL\mR^\tranH - \mZ_{\natural} , \mDelta\mR+ \mL\mUpsilon^\tranH\ra  \right| \\
		&\geq \frac{1}{2\lb 1+2\kappa\delta \rb\fronorm{\mL\mR^\tranH - \mZ_{\natural}}}\cdot\lb \frac{3}{4} - 51\kappa\delta \rb\fronorm{\mL\mR^\tranH - \mZ_{\natural}}^2\\
		&=\frac{3(1-68\kappa\delta)}{8(1+2\kappa\delta)}\fronorm{\mL\mR^\tranH - \mZ_{\natural}}^2\\
		&\geq \frac{1}{3} \fronorm{\mL\mR^\tranH - \mZ_{\natural}}^2,
	\end{align*}
	where the last line follows from the assumption $\delta\leq \frac{1}{(16+68\cdot 9)\kappa}$.

\subsection{Proof of Lemma \ref{upper bound of h}}
Let $\mZ_{k,\natural} = \mU_{k,\natural}\bSigma_{k,\natural}\mV_{k,\natural}$ be the compact SVD of $\mZ_{k,\natural}$, where $\mU_{k,\natural}\in\C^{sn_1\times r}$ and $\mV_{k,\natural}\in\C^{n_2\times r}$. The tangent space $T_k$ at $\mZ_{k,\natural}$ along with its orthogonal complement is defined as follows
\begin{align*}
	T_k :&= \left\{ \mZ_k\in\C^{sn_1\times n_2}~:~\mZ_k=\mU_{k,\natural} \mA^\tranH + \mB\mV_{k,\natural}^\tranH, \mA\in\C^{n_2\times r}, \mB\in\C^{sn_1\times r} \right\},\\
	T_k^\perp:&=\left\{ \lb \mI_{sn_1} - \mU_{k,\natural}\mU_{k,\natural}^\tranH \rb \mZ_k \lb \mI_{n_2} - \mV_{k,\natural} \mV_{k,\natural}^\tranH \rb~:~ 
	\mZ_k \in\C^{sn_1\times n_2} \right\}.
\end{align*}
For any block diagonal matrix $\mZ=\blkdiag(\mZ_1,\cdots, \mZ_K)$, we say $\mZ\in T$ if 
\begin{align}
	\label{tangent}
	T:=\left\{ \mZ: \mZ_k\in T_k, k=1,\cdots, K\right\}
\end{align}
and $\mZ\in T^\perp$ if 
\begin{align}
	\label{complement}
	T^\perp:= \left\{ \mZ: \mZ_k\in T_k^\perp, k=1,\cdots, K\right\}.
\end{align}
By the definition, one has
	\begin{align*}
		d_{\calS^\ast}(\nabla h(\mL), \nabla h(\mR)) 
		&=\fronorm{\begin{bmatrix}
				\lb \tG \lb \calA^\ast\calA-\calI \rb\tG^\ast(\mL\mR^\tranH - \mZ_{\natural}) \rb\mR(\mR^\tranH\mR)^{-\frac{1}{2}}\\
				\lb \tG \lb \calA^\ast\calA-\calI \rb\tG^\ast(\mL\mR^\tranH - \mZ_{\natural}) \rb^\tranH\mL(\mL^\tranH\mL)^{-\frac{1}{2}} \\
		\end{bmatrix}}\\
		&=\sup_{\mY} \la \begin{bmatrix}
			\lb \tG \lb \calA^\ast\calA-\calI \rb\tG^\ast(\mL\mR^\tranH - \mZ_{\natural}) \rb\mR(\mR^\tranH\mR)^{-\frac{1}{2}}\\
			\lb \tG \lb \calA^\ast\calA-\calI \rb\tG^\ast(\mL\mR^\tranH - \mZ_{\natural}) \rb^\tranH\mL(\mL^\tranH\mL)^{-\frac{1}{2}} \\
		\end{bmatrix}, \begin{bmatrix}
			\mY_1\\
			\mY_2\\
		\end{bmatrix}\ra\\
		&=\sup_{\mY} \underbrace{\la \tG \lb \calA^\ast\calA-\calI \rb\tG^\ast(\mL\mR^\tranH - \mZ_{\natural}), \mY_1(\mR^\tranH\mR)^{-\frac{1}{2}}\mR^\tranH+\mL(\mL^\tranH\mL)^{-\frac{1}{2}}\mY_2^\tranH\ra}_{:=I_2}, \numberthis\label{h1}
	\end{align*}
	where $\mY = \begin{bmatrix}
		\mY_1^\tran & \mY_2^\tran
	\end{bmatrix}^\tran$ such that $\fronorm{\mY}=1$. We will apply Lemma \ref{lemma inner product} to control $I_2$. Firstly, it can be seen that $I_2$ can be decomposed as follows:
	\begin{align*}
		\mL\mR^\tranH - \mZ_{\natural} &= (\mL_{\natural}+\mDelta)(\mR_{\natural}+\mUpsilon)^\tranH - \mL_{\natural}\mR_{\natural}^\tranH = \underbrace{\mL_{\natural}\mUpsilon^\tranH + \mDelta\mR_{\natural}^\tranH}_{:=\mPhi_1} + \underbrace{\mDelta\mUpsilon^\tranH }_{:=\mPsi_1},\\
		\mY_1(\mR^\tranH\mR)^{-\frac{1}{2}}\mR^\tranH+\mL(\mL^\tranH\mL)^{-\frac{1}{2}}\mY_2^\tranH&=\mY_1(\mR^\tranH\mR)^{-\frac{1}{2}}(\mR_\natural+ \mUpsilon)^\tranH+(\mL_{\natural}+\mDelta)(\mL^\tranH\mL)^{-\frac{1}{2}}\mY_2\\
		&=\underbrace{\mY_1(\mR^\tranH\mR)^{-\frac{1}{2}}\mR_\natural^\tranH + \mL_{\natural} (\mL^\tranH\mL)^{-\frac{1}{2}}\mY_2^\tranH}_{:=\mPhi_2} + \underbrace{\mY_1(\mR^\tranH\mR)^{-\frac{1}{2}}\mUpsilon^\tranH + \mDelta(\mL^\tranH\mL)^{-\frac{1}{2}}\mY_2}_{:=\mPsi_2}.
	\end{align*}
	According to the definition, it can be seen that $\mPhi_1, \mPhi_2\in T$ and $\mPsi_1, \mPsi_2\in T^\perp$. Therefore, Lemma \ref{lemma phi and psi} implies that
	\begin{align*}
		\fronorm{\mPhi_1} &=\fronorm{\mL_{\natural}\mUpsilon^\tranH +\mDelta\mR_{\natural}^\tranH}\leq (1+6\kappa\delta)\fronorm{\mL\mR^\tranH - \mZ_{\natural}},\\
		\fronorm{\mPsi_1}&=\fronorm{\mDelta\mUpsilon^\tranH} \leq \frac{6\kappa\delta}{\sqrt{K}} \fronorm{\mL\mR^\tranH - \mZ_{\natural}}.
	\end{align*}
	Furthermore, one has
	\begin{align*}
		\fronorm{\mPhi_2} &\leq \fronorm{\mY_1}\cdot \opnorm{\mR_{\natural}(\mR^\tranH\mR)^{-\frac{1}{2}}} + \fronorm{\mY_2}\cdot \opnorm{\mL_{\natural}(\mL^\tranH\mL)^{-\frac{1}{2}}}\leq \opnorm{\mR_{\natural}(\mR^\tranH\mR)^{-\frac{1}{2}}} +  \opnorm{\mL_{\natural}(\mL^\tranH\mL)^{-\frac{1}{2}}}
		\stackrel{(a)}{\leq} \frac{2}{1-2\kappa\delta},\\
		\fronorm{\mPsi_2}&\leq \fronorm{\mY_1}\cdot \opnorm{(\mR^\tranH\mR)^{-\frac{1}{2}}\mUpsilon^\tranH} + \opnorm{\mDelta(\mL^\tranH\mL)^{-\frac{1}{2}}}\cdot \fronorm{\mY_2}\leq \opnorm{(\mR^\tranH\mR)^{-\frac{1}{2}}\mUpsilon^\tranH} + \opnorm{\mDelta(\mL^\tranH\mL)^{-\frac{1}{2}}} 
		\stackrel{(b)}{\leq}\frac{1}{\sqrt{K}}\frac{4\kappa\delta}{1-2\kappa\delta},
	\end{align*}
	where step (a) and step (b) follows from Lemma \ref{fact 4}. Finally, by applying Lemma \ref{lemma inner product}, one has
	\begin{align*}
		I_2&=\la  \tG \lb \calA^\ast\calA-\calI \rb\tG^\ast\lb \mPhi_1 + \mPsi_1 \rb, \lb \mPhi_2 + \mPsi_2\rb\ra\\
		&\leq \epsilon \cdot \fronorm{\mPhi_1} \cdot  \fronorm{\mPhi_2}  + \sqrt{2K\mu_0s(1+\epsilon)}\cdot \lb \fronorm{\mPhi_1} \cdot \fronorm{\mPsi_2} +\fronorm{ \mPsi_1}\cdot \fronorm{\mPhi_2} \rb  + 2K\mu_0s \fronorm{\mPsi_1}\cdot \fronorm{\mPsi_2}\\
		&\leq \epsilon\cdot (1+6\kappa\delta)\fronorm{\mL\mR^\tranH - \mZ_{\natural}}\cdot \frac{2}{1-2\kappa\delta} \\
		&\quad + \sqrt{2K\mu_0s\cdot 2}\cdot \lb (1+6\kappa\delta)\fronorm{\mL\mR^\tranH - \mZ_{\natural}}\cdot \frac{1}{\sqrt{K}}\frac{4\kappa\delta}{1-2\kappa\delta}+\frac{6\kappa\delta}{\sqrt{K}}\fronorm{\mL\mR^\tranH-\mZ_{\natural}}\cdot \frac{2}{1-2\kappa\delta}\rb\\
		&\quad + 2K\mu_0s \cdot \frac{6\kappa\delta}{\sqrt{K}}\fronorm{\mL\mR^\tranH - \mZ_{\natural}}\cdot \frac{1}{\sqrt{K}}\frac{4\kappa\delta}{1-2\kappa\delta}\\
		&=\lb \frac{2\epsilon(1+6\kappa\delta)}{1-2\kappa\delta} + \sqrt{4K\mu_0s}\cdot\frac{1}{\sqrt{K}} \frac{4\kappa\delta(1+6\kappa\delta+3)}{1-2\kappa\delta} + 2K\mu_0s\cdot \frac{1}{K}\frac{24\kappa^2\delta^2}{1-2\kappa\delta} \rb \fronorm{\mL\mR^\tranH - \mZ_{\natural}}\\
		&=\lb \frac{2\epsilon(1+6\kappa\delta)}{1-2\kappa\delta} +  \frac{8\sqrt{\mu_0s}\kappa\delta(4+6\kappa\delta+6\sqrt{\mu_0s}\kappa\delta)}{1-2\kappa\delta} \rb \fronorm{\mL\mR^\tranH - \mZ_{\natural}}\\
		&\leq \lb \frac{2\epsilon(1+6\kappa\delta)}{1-2\kappa\delta} +  \frac{8\sqrt{\mu_0s}\kappa\delta(4+12\sqrt{\mu_0s}\kappa\delta)}{1-2\kappa\delta} \rb \fronorm{\mL\mR^\tranH - \mZ_{\natural}}\\
		&\leq \frac{2\epsilon+12\epsilon\kappa\delta+40\sqrt{\mu_0s}\kappa\delta}{1-2\kappa\delta}\fronorm{\mL\mR^\tranH - \mZ_{\natural}}\\
		&\leq \frac{2\epsilon+52\sqrt{\mu_0s}\kappa\delta}{1-2\kappa\delta}\fronorm{\mL\mR^\tranH - \mZ_{\natural}}\\
		&\leq (3\epsilon + 58\sqrt{\mu_0s}\kappa \delta)\fronorm{\mL\mR^\tranH - \mZ_{\natural}} \\
		&\leq \frac{1}{5}\fronorm{\mL\mR^\tranH - \mZ_{\natural}}, \numberthis\label{h2}
	\end{align*}
	where we have used the fact that $\epsilon\leq \frac{1}{15}, \mu_0s\geq 1$ and $\delta\leq \frac{1}{290\sqrt{\mu_0s}\kappa}$, that is, $1/(1-2\delta) \leq \frac{3}{2} $. Since \eqref{h2} holds for any $\mY$ such that $\fronorm{\mY}=1$, thus plugging \eqref{h2} into \eqref{h1}, we complete the proof.

\subsection{Proof of Lemma \ref{lemma bound L}}
According to the definition of $L$ in \eqref{smooth parameter}, it can be seen that $L\geq 4$. Thus we focus on providing an upper bound of $L$. Firstly, we provide an upper bound of $d_{\calS}(\mV,\mW)$ for $\mV=-\eta \nabla f(\mL)(\mR^\tranH\mR)^{-1}$ and $\mW = -\eta\nabla f(\mR)(\mL^\tranH\mL)^{-1}$. Recall \eqref{gradient f L} and \eqref{gradient f R} that
\begin{align*}
	\nabla f(\mL) 
	&=\underbrace{\lb \mL\mR^\tranH - \mZ_{\natural} \rb \mR}_{=\nabla g(\mL)}+\underbrace{ \lb \tG \lb \calA^\ast\calA-\calI \rb\tG^\ast(\mL\mR^\tranH - \mZ_{\natural}) \rb\mR}_{=\nabla h(\mL)},\\
	\nabla f(\mR) 
	&=\underbrace{\lb \mL\mR^\tranH - \mZ_{\natural} \rb^\tranH \mL }_{=\nabla g(\mR)}+ \underbrace{\lb \tG\lb \calA^\ast\calA-\calI \rb\tG^\ast(\mL\mR^\tranH - \mZ_{\natural})  \rb^\tranH\mL}_{=\nabla h(\mR)}.
\end{align*}
By the definition of $d_{\calS}(\mV,\mW)$ in \eqref{def ds}, we have
\begin{align*}
	d_{\calS}^2(-\eta \nabla f(\mL)(\mR^\tranH\mR)^{-1},-\eta\nabla f(\mR)(\mL^\tranH\mL)^{-1})
	&=\eta^2\lb \fronorm{\nabla f(\mL)(\mR^\tranH\mR)^{-\frac{1}{2}}}^2 + \fronorm{\nabla f(\mR)(\mL^\tranH\mL)^{-\frac{1}{2}}}^2 \rb\\
	&\leq 2\eta^2  \lb \fronorm{\nabla g(\mL)(\mR^\tranH\mR)^{-\frac{1}{2}}}^2 + \fronorm{\nabla g(\mR)(\mL^\tranH\mL)^{-\frac{1}{2}}}^2 \rb \\
	&\quad + 2\eta^2 \lb \fronorm{\nabla h(\mL)(\mR^\tranH\mR)^{-\frac{1}{2}}}^2 + \fronorm{\nabla h(\mR)(\mL^\tranH\mL)^{-\frac{1}{2}}}^2 \rb \\
	&\stackrel{(a)}{=}2\eta^2 d_{\calS^\ast}^2(\nabla g(\mL), \nabla g(\mR)) + 2\eta^2 d_{\calS^\ast}^2(\nabla h(\mL), \nabla h(\mR))\\
	&\stackrel{(b)}{\leq}2\eta^2  \cdot 2\fronorm{\mL\mR^\tranH - \mZ_{\natural}}^2+ 2\eta^2 \cdot (3\epsilon+58\kappa\delta\sqrt{\mu_0s})^2 \fronorm{\mL\mR^\tranH - \mZ_{\natural}}^2\\
	&=\lb 4+ 2(3\epsilon+58\kappa\delta\sqrt{\mu_0s})^2\rb\eta^2 \fronorm{\mL\mR^\tranH - \mZ_{\natural}}^2\\
	&\stackrel{(c)}{\leq } 5\eta^2\fronorm{\mL\mR^\tranH - \mZ_{\natural}}^2, 
\end{align*}
where step (a) is due to the definition of $d_{\calS^\ast}$ in \eqref{def dual ds}, step (b) is due to Lemma \ref{bound of nabla g} and Lemma \ref{upper bound of h}, step (c) follows from that $\epsilon\leq \frac{1}{6\sqrt{2}}$ and $\delta\leq \frac{1}{116\kappa\sqrt{2\mu_0s}}$, i.e., $3\epsilon+58\kappa\delta\sqrt{\mu_0s} \leq \frac{1}{\sqrt{2}}$.

Secondly, we provide an upper bound of $L$. Since $\fronorm{\mL_k\mR_k^\tranH - \mZ_{k, \natural}}\leq \frac{\delta}{\sqrt{K}}\sigma_r(\mZ_{k, \natural})$ for all $k\in[K]$, one has
\begin{align*}
	\fronorm{\mL\mR^\tranH - \mZ_{\natural}} &=\sqrt{\sum_{k=1}^{K}\fronorm{\mL_k\mR_k^\tranH - \mZ_{k, \natural}}^2} \leq \sqrt{\frac{\delta^2}{K}\sum_{k=1}^{K}\sigma_r^2(\mZ_{k, \natural})} \leq \delta\cdot \max_k \sigma_r(\mZ_{k, \natural}) \leq \delta \cdot \max_k \sigma_1(\mZ_{k, \natural}).
\end{align*}
Therefore, we have
\begin{align*}
	\frac{d_{\calS}(\mV, \mW)}{\min_k \sigma_r(\mZ_{k, \natural})} &\leq  \frac{\sqrt{5}\eta \fronorm{\mL\mR^\tranH - \mZ_{\natural}}}{\min_k \sigma_r(\mZ_{k, \natural})} \leq \frac{3\eta \delta \cdot  \max_k \sigma_1(\mZ_{k, \natural})}{\min_k \sigma_r(\mZ_{k, \natural})} =3\kappa\delta\eta.
\end{align*}
Thus
\begin{align*}
	L&=3+ \frac{2\fronorm{\mL\mR^\tranH - \mZ_{\natural} }}{\min_k \sigma_r(\mZ_{k, \natural})}  +\frac{3d_{\calS}^2(\mV, \mW)}{\min_k \sigma_r^2(\mZ_{k, \natural})} \\
	&\leq 3+ \frac{2\fronorm{\mL\mR^\tranH - \mZ_{\natural} }}{\min_k \sigma_r(\mZ_{k, \natural})}  +\frac{3\cdot 5\eta^2 \fronorm{\mL\mR^\tranH - \mZ_{\natural}}^2}{\min_k \sigma_r^2(\mZ_{k, \natural})} \\ 
	&\leq 3 + 2\delta\kappa + 15\eta^2\delta^2\kappa^2 \\
	&\leq 5.
\end{align*}
where the last line is due to $\eta\leq 1$, $\delta\leq \frac{1}{4\kappa}$, i.e., $2\delta\kappa\leq \frac{1}{2}\leq 1$ and $15\eta^2\delta^2\kappa^2\leq \frac{15}{16}\leq 1$.

\subsection{Proof of Lemma \ref{lemma linear convergence}}

According to Lemma \ref{smooth of g}, one has
\begin{align*}
	g(\mL_{t+1}, \mR_{t+1}) &\leq g(\mL_t,\mR_t) -\eta_t \la \nabla g(\mL_t), \nabla f(\mL_t)(\mR_t^\tranH\mR_t)^{-1}\ra -\eta_t\la\nabla g(\mR_t), \nabla f(\mR_t)(\mL_t^\tranH\mL_t)^{-1} \ra \\
	&\quad + \frac{L}{2}d_{\calS_t}^2\lb -\eta_t\nabla (\mL_t)(\mR_t^\tranH\mR_t)^{-1}, -\eta_t\nabla f(\mR_t)(\mL_t^\tranH\mL_t)^{-1} \rb\\
	&=g(\mL_t,\mR_t) -\eta_t \la \nabla g(\mL_t), \lb \nabla g(\mL_t) + \nabla h(\mL_t) \rb(\mR_t^\tranH\mR_t)^{-1}\ra \\
	&\quad -\eta_t\la\nabla g(\mR_t), \lb \nabla g(\mR_t)+\nabla h(\mR_t) \rb(\mL_t^\tranH\mL_t)^{-1} \ra \\
	&\quad + \frac{L}{2}d_{\calS_t}^2\lb -\eta_t\nabla f(\mL_t)(\mR_t^\tranH\mR_t)^{-1}, -\eta_t\nabla f(\mR_t)(\mL_t^\tranH\mL_t)^{-1} \rb\\
	&=g(\mL_t,\mR_t) - \eta_t \lb \fronorm{\nabla g(\mL_t)(\mR_t^\tranH\mR_t)^{-\frac{1}{2}} }^2 + \fronorm{ \nabla g(\mR_t)(\mL_t^\tranH\mL_t)^{-\frac{1}{2}} }^2\rb\\
	&\quad -\eta_t  \la \nabla g(\mL_t)(\mR_t^\tranH\mR_t)^{-\frac{1}{2}}, \nabla h(\mL_t)(\mR_t^\tranH\mR_t)^{-\frac{1}{2}}\ra -\eta_t  \la \nabla g(\mR_t)(\mL_t^\tranH\mL_t)^{-\frac{1}{2}}, \nabla h(\mR_t)(\mL_t^\tranH\mL_t)^{-\frac{1}{2}}\ra\\
	&\quad +\frac{L\eta_t^2}{2}  \lb \fronorm{\lb \nabla g(\mL_t)+ \nabla h(\mL_t)\rb(\mR_t^\tranH\mR_t)^{-\frac{1}{2}} }^2 + \fronorm{ \lb \nabla g(\mR_t)+\nabla h(\mR_t)\rb(\mL_t^\tranH\mL_t)^{-\frac{1}{2}} }^2\rb\\
	&\stackrel{(a)}{\leq }g(\mL_t,\mR_t) - \eta_t \lb \fronorm{\nabla g(\mL_t)(\mR_t^\tranH\mR_t)^{-\frac{1}{2}} }^2 + \fronorm{ \nabla g(\mR_t)(\mL_t^\tranH\mL_t)^{-\frac{1}{2}} }^2\rb\\
	&\quad +\frac{\eta_t}{4}  \fronorm{\nabla g(\mL_t)(\mR_t^\tranH\mR_t)^{-\frac{1}{2}}}^2 + \eta_t\fronorm{ \nabla h(\mL_t)(\mR_t^\tranH\mR_t)^{-\frac{1}{2}} }^2\\
	&\quad +\frac{\eta_t}{4} \fronorm{\nabla g(\mR_t)(\mL_t^\tranH\mL_t)^{-\frac{1}{2}}}^2 + \eta\fronorm{ \nabla h(\mR_t)(\mL_t^\tranH\mL_t)^{-\frac{1}{2}} }^2\\
	&\quad +L\eta_t^2 \lb \fronorm{ \nabla g(\mL_t)(\mR_t^\tranH\mR_t)^{-\frac{1}{2}} }^2 + \fronorm{  \nabla h(\mR_t)(\mL_t^\tranH\mL_t)^{-\frac{1}{2}} }^2\rb\\
	&\quad +L\eta_t^2 \lb \fronorm{ \nabla h(\mL_t)(\mR_t^\tranH\mR_t)^{-\frac{1}{2}} }^2 + \fronorm{  \nabla h(\mR_t)(\mL_t^\tranH\mL_t)^{-\frac{1}{2}} }^2\rb\\
	&= g(\mL_t, \mR_t) - \lb \frac{3\eta_t}{4} - L\eta_t^2\rb d_{\calS_t^\ast}^2(\nabla g(\mL_t), \nabla g(\mR_t)) + \lb \eta_t + L\eta_t^2\rb d_{\calS_t^\ast}^2(\nabla h(\mL_t), \nabla h(\mR_t))\\
	&\stackrel{(b)}{\leq}g(\mL_t, \mR_t) - \frac{\eta_t}{2} d_{\calS_t^\ast}^2(\nabla g(\mL_t), \nabla g(\mR_t)) + \frac{5\eta_t}{4} d_{\calS_t^\ast}^2(\nabla h(\mL_t), \nabla h(\mR_t))\\
	&\stackrel{(c)}{\leq } g(\mL_t, \mR_t) - \frac{\eta_t}{2} \cdot \frac{1}{9}\fronorm{\mL_t\mR_t^\tranH-\mZ_{\natural}}^2+ \frac{5\eta_t}{4} \cdot \frac{1}{25} \fronorm{\mL_t\mR_t^\tranH-\mZ_{\natural}}^2\\
	&=g(\mL_t, \mR_t)  - \frac{\eta_t}{180}  \fronorm{\mL_t\mR_t^\tranH-\mZ_{\natural}}^2\\
	&=\lb 1-\frac{\eta}{90}\rb g(\mL_t, \mR_t),
\end{align*}
where step (a) is due to $\fronorm{\mA+\mB} \leq 2\fronorm{\mA}^2 + 2\fronorm{\mB}^2$ and $-\la \mA, \mB\ra\geq \frac{1}{4}\fronorm{\mA}^2 + \fronorm{\mB}^2 $, step (b) follows from that $\eta_t\leq \frac{1}{20}$ and $L\leq 5$ (Lemma \ref{lemma bound L}), i.e., $\frac{3\eta_t}{4}-L\eta_t^2\geq \frac{\eta_t}{2}$ and $\eta_t + L\eta_t^2\leq \frac{5\eta_t}{4}$, step (c) follows from Lemma \ref{bound of nabla g} and Lemma \ref{upper bound of h}. 

\subsection{Proof of Lemma \ref{theorem: initial}}
\label{proof initial}
Recall that the initialization in Algorithm \ref{alg: SGD} is given by
\begin{align*}
	\mZ_{\ell,0} = \calP_r\calH\calA_\ell^\ast (\vy), ~ \ell=1,\cdots, K,
\end{align*}
where $\vy = \sum_{k=1}^{K} \calA_k(\mX_k)$.  Let $\mZ_{\ell, 0} =\mU_{\ell, 0} \bSigma_{\ell, 0}\mV_{\ell, 0}^\tranH$ be the SVD of $\mZ_{\ell, 0}$. Denote $\mL_{\ell,0} = \mU_{\ell, 0}\bSigma_{\ell, 0}^{\frac{1}{2}}, \mR_{\ell, 0} = \mV_{\ell, 0}\bSigma_{\ell, 0}^{\frac{1}{2}}$ and $\mW_{\ell,0} = \calH\calA_\ell^\ast (\vy)$. A simple computation yields that
\begin{align*}
	\mW_{k,0} &=\calG\calD\calA_\ell^\ast(\vy) = \calG\calA_\ell^\ast (\mD\vy) =\calG\calA_\ell^\ast\lb \sum_{k=1}^{K} \calA_k\calG(\mZ_{k,\natural})\rb =\sum_{k=1}^{K}\calG\calA_\ell^\ast\calA_k\calGT(\mZ_{k, \natural})\\
	&=\calG\calA_\ell^\ast\calA_\ell\calGT(\mZ_{\ell, \natural}) + \sum_{k\neq \ell}\calG\calA_\ell^\ast\calA_k\calGT(\mZ_{k, \natural}),
\end{align*}
Thus one has
\begin{align*}
	\opnorm{\mZ_{\ell, 0} - \mZ_{\ell, \natural}} &=\opnorm{ \calP_r(\mW_{\ell,0}) - \mZ_{\ell, \natural}}\\
	&\leq \opnorm{ \calP_r(\mW_{\ell,0}) - \mW_{\ell,0}} + \opnorm{\mW_{\ell,0} - \mZ_{\ell, \natural}}\\
	&\stackrel{(a)}{\leq} \opnorm{ \mZ_{\ell, \natural}- \mW_{\ell,0}} + \opnorm{\mW_{\ell,0} - \mZ_{\ell, \natural}}\\
	&=2\opnorm{\calG\calA_\ell^\ast\calA_\ell\calGT(\mZ_{\ell, \natural}) + \sum_{k\neq \ell}\calG\calA_\ell^\ast\calA_k\calGT(\mZ_{k, \natural}) - \mZ_{\ell, \natural}}\\
	&\leq \underbrace{2\opnorm{\calG\calA_\ell^\ast\calA_\ell\calGT(\mZ_{\ell, \natural}) - \mZ_{\ell, \natural}} }_{:=I_6}+ \underbrace{2\opnorm{\sum_{k\neq \ell}\calG\calA_\ell^\ast\calA_k\calGT(\mZ_{k, \natural}) } }_{:=I_7},
\end{align*}
where step (a) follows from that $\calP_r(\mW_{\ell,0})$ is the best rank-$r$ approximation of $\mW_{\ell,0}$. To this end, we bound $I_6$ and $I_7$, respectively.
\begin{itemize}
	\item Bounding of $I_6$. According to Lemma \ref{lemma initial 1}, one has
	\begin{align*}
		I_1 &=2\opnorm{\calG\calA_\ell^\ast\calA_\ell\calGT(\mZ_{\ell, \natural}) - \mZ_{\ell, \natural}} \\
		&\leq 2c_1\sqrt{\frac{\mu_0s \mu_1 r\log(sn)}{n}}\sigma_1(\mZ_{\ell, \natural})=2c_1\sqrt{\frac{\kappa_\ell^2\mu_0 s\mu_1 r\log(sn)}{n}}\sigma_r(\mZ_{\ell, \natural})\\
		&\leq \frac{\epsilon}{2}\sigma_1(\mZ_{\ell, \natural})\numberthis\label{eq initial 1}
	\end{align*}
	with probability at least $1-(sn)^{-\gamma+1}$, where $\kappa_\ell = \sigma_1(\mZ_{\ell,\natural})/\sigma_r(\mZ_{\ell,\natural})$, and the last line is due to $n\geq C_{\gamma} \epsilon^{-2}\kappa^2_\ell\mu_0 s\mu_1 r\log(sn)$.
	
	\item Bounding of $I_7$. Denote $\vv_{\setminus\ell}[i] = \sum_{k\neq \ell} \la \vb_{k,i}\ve_i^\tran, \calGT(\mZ_{k,\natural}) \ra$ for any $i\in [n]$. A simple computation yields that 
	\begin{align*}
		I_7 &=2\opnorm{\sum_{k\neq \ell}\calG\calA^\ast_\ell\calA_k\calGT(\mZ_{k, \natural})} =2\opnorm{\calG\calA_\ell^\ast(\vv_{\setminus\ell})} = 2\opnorm{\sum_{i=1}^{n} \vv_{\setminus\ell}[i] \calG(\vb_{\ell, i}\ve_i^\tran)}:=2\opnorm{\sum_{i=1}^{n} \mY_i},
	\end{align*}
	where $\mY_i = \vv_{\setminus\ell}[i] \calG(\vb_{\ell, i} \ve_i^\tran)$. Notice that $\{\mY_i\}_{i=1}^n$ are independent random matrices with 
	\begin{align*}
		\E{\mY_i} &= \E{ \vv_{\setminus\ell}[i] \calG(\vb_{\ell, i} \ve_i^\tran) } =\E{\sum_{k\neq \ell} \la \vb_{k,i}\ve_i^\tran, \calGT(\mZ_{k,\natural})\ra \calG(\vb_{\ell, i}\ve_i^\tran)}\\
		&=\E{\sum_{k\neq \ell} \calG\lb \vb_{\ell,i} \vb_{k,i}^\tranH\calGT(\mZ_{k,\natural}\ve_i) \ve_i^\tran\rb}\\
		&=\sum_{k\neq \ell} \calG\lb \E{\vb_{\ell,i} \vb_{k,i}^\tranH}\calGT(\mZ_{k,\natural}\ve_i) \ve_i^\tran\rb\\
		&=0,
	\end{align*}
	where the last line has used the fact that $\vb_{\ell,i} $ is independent of $\vb_{k,i}$ for $\ell\neq k$ and $\E{\vb_{\ell,i}}=\bzero$ (Assumption \ref{assumption 1}). Moreover, one has
	\begin{align*}
		\opnorm{\mY_i}&\leq |\vv_\ell[i]| \cdot \opnorm{\calG(\vb_{\ell, i}\ve_i^\tran)} \\
		&\stackrel{(a)}{\leq} \sqrt{\mu_0s}\sum_{k\neq \ell} \twonorm{\calGT(\mZ_{k, \natural})\ve_i}\cdot \opnorm{\mG_i\otimes \vb_{\ell, i}}\\
		&\leq \sqrt{\mu_0s}\sum_{k\neq \ell} \twonorm{\calGT(\mZ_{k, \natural})\ve_i}\cdot \frac{\sqrt{\mu_0s}}{\sqrt{w_i}}\\
		&\leq \mu_0s \sum_{k\neq \ell} \max_i\frac{\twonorm{\calGT(\mZ_{k, \natural}\ve_i)}}{\sqrt{w_i}}\\
		&\stackrel{(b)}{\leq }\mu_0s \sum_{k\neq \ell} \frac{\mu_1 r}{n}\sigma_1(\mZ_{k, \natural})\\
		&\leq \frac{\mu_0s\mu_1 r}{n} \sum_{k=1}^{K}\sigma_1(\mZ_{k, \natural}),
	\end{align*}
	where step (a) is due to the fact $|\vv_{\setminus\ell}[i]| \leq  \sqrt{\mu_0s}\sum_{k\neq \ell} \twonorm{\calGT(\mZ_{k, \natural})\ve_i}$ ,  step (b) is due to Lemma \ref{lemma v6}. The fact can be proved as follows:
	\begin{align*}
		|\vv_{\setminus\ell}[i]| &\leq \sum_{k\neq \ell} \left| \vb_{k,i}^\tranH \calGT(\mZ_{k, \natural})\ve_i\right|\\
		&\leq \sum_{k\neq \ell} \twonorm{ \vb_{k,i}}\cdot \twonorm{\calGT(\mZ_{k, \natural})\ve_i}\\
		&\leq \sqrt{\mu_0s}\sum_{k\neq \ell} \twonorm{\calGT(\mZ_{k, \natural})\ve_i}.
	\end{align*}
	Furthermore, by the definition of $\vv_\ell$, it can be seen that $\vv_{\ell}$ is independent of $\calG(\vb_{\ell, i})$. Thus one has
	\begin{align*}
		\E{|\vv_{\ell}[i]|^2} &=\E{ \left| \sum_{k\neq \ell} \la \vb_{k,i}\ve_i^\tran, \calGT(\mZ_{k,\natural}) \ra \right|^2}\\
		&=\sum_{k\neq \ell}\E{ \left|  \la \vb_{k,i}\ve_i^\tran, \calGT(\mZ_{k,\natural}) \ra \right|^2}\\
		&\leq \sum_{k\neq \ell} \mu_0 s\twonorm{\calGT(\mZ_{k, \natural})\ve_i}^2
	\end{align*}
	and
	\begin{align*}
		\opnorm{\sum_{i=1}^{n}\E{\mY_i\mY_i^\tranH}} &\leq \opnorm{\sum_{i=1}^{n}\E{|\vv_\ell[i]|^2}\cdot \E {\calG(\vb_{\ell, i}\ve_i^\tran)\calG(\vb_{\ell, i}\ve_i^\tran)^\tranH}}\\
		&\leq \sum_{i=1}^{n} \E{|\vv_\ell[i]|^2}\cdot \opnorm{ \E {\calG(\vb_{\ell, i}\ve_i^\tran)\calG(\vb_{\ell, i}\ve_i^\tran)^\tranH}}\\
		&=\sum_{i=1}^{n} \E{|\vv_\ell[i]|^2}\cdot \opnorm{ \mG_i\mG_i^\tranH \otimes \E{\vb_{\ell, i}\vb_{\ell, i}^\tranH}}\\
		&\leq \sum_{i=1}^{n} \E{|\vv_\ell[i]|^2}\cdot \frac{1}{w_i}\\
		&\leq \sum_{i=1}^{n} \sum_{k\neq \ell} \mu_0 s\twonorm{\calGT(\mZ_{k, \natural})\ve_i}^2\cdot \frac{1}{w_i}\\
		&=\mu_0s\sum_{k\neq \ell}\sum_{i=1}^{n} \frac{\twonorm{\calGT(\mZ_{k, \natural})\ve_i}^2}{w_i}\\
		&\leq \mu_0s \sum_{k=1}^{K}  \sum_{i=1}^{n} \frac{\twonorm{\calGT(\mZ_{k, \natural})\ve_i}^2}{w_i} \\
		&\leq \frac{\mu_0s  \mu_1 r\log(sn)}{n}\sum_{k=1}^{K}\sigma_1^2(\mZ_{k, \natural}),
	\end{align*}
	where the last line is due to Lemma \ref{lemma v6}. Similarly, one can obtain the same upper bound for $\opnorm{\sum_{i=1}^{n}\E{\mY_i^\tranH\mY_i}}$. Applying matrix Bernstein inequality shows that, with probability at least $1-(sn)^{-\gamma+1}$,
	\begin{align*}
		\opnorm{\sum_{k\neq \ell}\calG\calA^\ast_\ell\calA_k\calGT(\mZ_{k, \natural})} &\leq c\lb  \frac{\mu_0s\mu_1 r\log (sn)}{n} \sum_{k=1}^{K}\sigma_1(\mZ_{k, \natural}) + \sqrt{ \frac{\mu_0s  \mu_1 r\log^2(sn)}{n}\sum_{k=1}^{K}\sigma_1^2(\mZ_{k, \natural}) } \rb\\
		&\leq c\lb \frac{K\mu_0s\mu_1 r\log (sn)}{n} \max_k\sigma_1(\mZ_{k, \natural}) + \sqrt{ \frac{K\mu_0s  \mu_1 r\log^2(sn)}{n}}\max_k\sigma_1(\mZ_{k, \natural}) \rb\\
		&=c\lb \frac{K\mu_0s\mu_1 r\log (sn)}{n} \max_k\sigma_1(\mZ_{k, \natural}) + \sqrt{ \frac{K\mu_0s  \mu_1 r\log^2(sn)}{n}}\max_k\sigma_1(\mZ_{k, \natural}) \rb \\
		&=c\lb \frac{K\kappa\mu_0s\mu_1 r\log (sn)}{n} \min_k \sigma_r(\mZ_{k,\natural})  + \sqrt{ \frac{K\kappa^2\mu_0s  \mu_1 r\log^2(sn)}{n}}\min_k \sigma_r(\mZ_{k,\natural}) \rb\\
		&\leq \frac{\epsilon}{2}\sigma_r(\mZ_{\ell, \natural}), \numberthis \label{eq initial 2}
	\end{align*}
	where $\kappa= \frac{\max_k\sigma_1(\mZ_{k, \natural}) }{\min_k \sigma_r(\mZ_{k,\natural})}$, and the last line is due to $n\geq C_{\gamma}\epsilon^{-2} K\mu_0s\mu_1r\kappa^2\log^2(sn)$.
\end{itemize}
Combing \eqref{eq initial 1} and \eqref{eq initial 2} together, we have
\begin{align*}
	\opnorm{\mZ_{\ell, 0} - \mZ_{\ell, \natural}} &\leq I_1 + I_2\leq \epsilon \sigma_r(\mZ_{\ell, \natural})
\end{align*}
happens with probability at least $1-(sn)^{-\gamma+1}$ provided that $n\geq C_{\gamma} \epsilon^{-2}K\mu_0s\mu_1r\kappa^2\log^2(sn)$. By taking union bound over $1\leq \ell\leq K$, 
\begin{align*}
	\Pr{\opnorm{\mZ_{\ell, 0} - \mZ_{\ell, \natural}} \leq\epsilon \sigma_r(\mZ_{\ell, \natural}) ~ \forall \ell =1,\cdots, K} \geq 1-K(sn)^{-\gamma+1}.
\end{align*}
Since $\mZ_{\ell,0} - \mZ_{\ell,\natural}$ is rank-$2r$ matrix, it implies that $\fronorm{\mZ_{\ell, 0} - \mZ_{\ell, \natural} } \leq \sqrt{2r}\opnorm{\mZ_{\ell, 0} - \mZ_{\ell, \natural}}$. 
Letting $\epsilon=\frac{\rho}{\sqrt{Ksr}}$ and $\gamma = \gamma'+\log K$, or $n\geq C_{\gamma'+\log K}\rho^{-2} K^2s^2r^2\mu_0\mu_1\kappa^2\log^2(sn)$, 
\begin{align*}
	\Pr{\fronorm{\mZ_{\ell, 0} - \mZ_{\ell, \natural}} \leq\frac{\rho}{\sqrt{Ks}}\sigma_r(\mZ_{\ell, \natural}) ~ \forall \ell =1,\cdots, K} \geq 1-(sn)^{-\gamma'+1},
\end{align*}
which completes the proof.

\section{Technical lemmas}

\begin{lemma}\cite[Corollary III.11]{chen2022vectorized}
	\label{lemma local RIP}
	Suppose that $n\geq C_{\gamma} \epsilon^{-2} \mu_0s \mu_1r\log(sn)$. Then with probability at least $1-(sn)^{-\gamma+1}$, there holds the following inequality
	\begin{align*}
		\opnorm{\calP_{T_k} \calG(\calA_k^\ast \calA_k - \calI)\calG\calP_{T_k}} \leq \epsilon,
	\end{align*}
	where $C_\gamma$ is an absolute constant only depending on $\gamma\geq 2$.
\end{lemma}
\begin{lemma}\label{cross}
	Suppose that $n\geq C_\gamma \epsilon^{-2}K^2\mu_0s\mu_1r\log(sn)$. Thenwith probability at least $1-(sn)^{-\gamma+1}$, there holds the following inequality
	\begin{align*}
		\opnorm{\calP_{T_k}\calG\calA_k^\ast \calA_j\calG^\ast\calP_{T_j}} \leq \frac{\epsilon}{K}\quad 1\leq k\neq j\leq K.
	\end{align*}
\end{lemma}

\begin{proof}
	For any fixed $\mZ\in\C^{sn_1\times n_2}$, by the definitions of $\calA_k$ and $\calA^\ast_j$ in \eqref{A} and \eqref{At}, we have
	\begin{align*}
		\calP_{T_k}\calG\calA_{k}^\ast \calA_{j}\calG^\ast\calP_{T_j}(\mZ) &= \sum_{i=1}^n \calP_{T_k}\calG(\vb_{k,i} \ve_{i}^\tran) \la \vb_{j, i}\ve_i^\tran, \calGT \calP_{T_j}(\mZ) \ra\\
		&=\sum_{i=1}^n \calP_{T_k}\calG(\vb_{k,i} \ve_{i}^\tran) \la \calP_{T_j}\calG(\vb_{j, i}\ve_i^\tran), \mZ \ra.
	\end{align*}

Let $\vx_{k,i} =\text{vec}( \calP_{T_k}\calG(\vb_{k,i} \ve_{i}^\tran) )$ and $\vz = \text{vec}(\mZ)$. According to Assumption \ref{assumption 1}, $\vx_{k,i}$ is independent of $\vx_{j,i}$ for any $k\neq j$. A direct computation yields that 
	\begin{align*}
		\opnorm{\calP_{T_k}\calG\calA_{k,t}^\ast \calA_{j,t}\calG^\ast\calP_{T_j}} &=\max_{\fronorm{\mZ}=1} \fronorm{ \calP_{T_k}\calG\calA_{k,t}^\ast \calA_{j,t}\calG^\ast\calP_{T_j}(\mZ)}\\
		&=\max_{\fronorm{\mZ}=1} \fronorm{\sum_{i=1}^n  \vx_{k,i} \vx_{j,i}^\tranH \vz }\\
		&=\opnorm{\sum_{i=1}^n  \vx_{k,i} \vx_{j,i}^\tranH},
	\end{align*}
where it can be seen that $\{\vx_{k,i}\vx_{j,i}\}_{i=1,\cdots, n}$ are mean-zero independent random matrices. Firstly, $\opnorm{ \vx_{k,i} \vx_{j,i}^\tranH } $ can be bounded as follows:
	\begin{align*}
		\opnorm{ \vx_{k,i} \vx_{j,i}^\tranH } &\leq \twonorm{\vx_{k,i}} \cdot \twonorm{ \vx_{j,i}^\tranH} \\
		&=\fronorm{\calP_{T_k}\calG(\vb_{k,i} \ve_{i}^\tran)  } \cdot \fronorm{ \calP_{T_j}\calG(\vb_{j,i} \ve_{i}^\tran) }\\
		&\leq  \frac{2\mu_0\mu_1sr}{n},
	\end{align*}
	where the last inequality follows from Lemma VI.2 in \cite{chen2022vectorized}.  Secondly, 
	\begin{align*}
		\opnorm{\sum_{i=1}^n \mathbb{E}\left [\vx_{k,i} \vx_{j,i}^\tranH \vx_{j,i}\vx_{k,i}^\tranH \right]}  &= \opnorm{\sum_{i=1}^n\mathbb{E}\left [\twonorm{\vx_{j,i}}^2\vx_{k,i} \vx_{k,i}^\tranH \right] } \\
		& = \opnorm{\sum_{i=1}^n \E{\twonorm{\vx_{j,i}}^2}\E{\vx_{k,i} \vx_{k,i}^\tranH}}\\
		&\leq \max_{i}\E{\twonorm{\vx_{j,i}}^2}\opnorm{\sum_{i=1}^n \mathbb{E}\left [\vx_{k,i} \vx_{k,i}^\tranH \right]}\\
		&\leq \frac{2\mu_0\mu_1sr}{n}\cdot\opnorm{\sum_{i=1}^n \mathbb{E}\left [\vx_{k,i} \vx_{k,i}^\tranH \right]}\\
		&\leq \frac{2\mu_0\mu_1sr}{n},
	\end{align*}
where the last line is due to the fact $\opnorm{\sum_{i=1}^n \mathbb{E}\left [\vx_{k,i} \vx_{k,i}^\tranH \right]}\leq 1$. The fact can be proved as follows:
	\begin{align*}
		\opnorm{\sum_{i=1}^n \mathbb{E}\left [\vx_{k,i} \vx_{k,i}^\tranH \right]} &=\sup_{\fronorm{\mW}=1} \twonorm{ \sum_{i=1}^n \E{\vx_{k,i} \vx_{k,i}^\tranH \text{vec}(\mW)} }\\
		&=\sup_{\fronorm{\mW}=1} \fronorm{\E{\sum_{i=1}^{n} \calP_{T_k}\calG(\vb_{k,i} \ve_i^\tranH) \la\calP_{T_k}\calG(\vb_{k,i} \ve_i^\tranH), \mW\ra} }\\
		&=\sup_{\fronorm{\mW}=1} \fronorm{\E{\sum_{i=1}^{n} \calP_{T_k}\calG\left(\vb_{k,i} \vb_{k,i}^\tranH \calGT\calP_{T_k}(\mW)  \ve_i \ve_i^\tranH \right)}}\\
		&=\sup_{\fronorm{\mW}=1} \fronorm{ \sum_{i=1}^{n} \calP_{T_k}\calG\left(\E{\vb_{k,i} \vb_{k,i}^\tranH} \calGT\calP_{T_k}(\mW)  \ve_i \ve_i^\tranH \right)}\\
		&=\sup_{\fronorm{\mW}=1} \fronorm{ \sum_{i=1}^{n} \calP_{T_k}\calG\left( \calGT\calP_{T_k}(\mW)  \ve_i \ve_i^\tranH \right)}\\
		&=\sup_{\fronorm{\mW}=1} \fronorm{  \calP_{T_k}\calG  \calGT\calP_{T_k}(\mW) }\\
		&\leq 1.
	\end{align*}
	Using the same argument, one can obtain the same upper bound for $\opnorm{\sum_{i=1}^n \mathbb{E}\left [\vx_{j,i} \vx_{k,i}^\tranH \vx_{k,i}\vx_{j,i}^\tranH \right]}$.
	Then we apply the matrix Bernstein inequality to obtain an upper bound of $\opnorm{\sum_{i=1}^n  \vx_{k,i} \vx_{j,i}^\tranH} $ for fixed $k$ and $j$ such that $k\neq j$, 
	\begin{align*}
		\opnorm{\sum_{i=1}^n  \vx_{k,i} \vx_{j,i}^\tranH} &\leq  C\left(
		\frac{\gamma\mu_0\mu_1sr\log(sn)}{n}+\sqrt{\frac{\gamma\mu_0\mu_1sr\log(sn)}{n}} \right)\leq \frac{\epsilon}{K}
	\end{align*}
with probability at least $1-(sn)^{-\gamma+1}$for any $\gamma\geq 2$ provided that $n\geq C_{\gamma}\epsilon^{-2}K^2 \mu_0\mu_1sr\log(sn)$.  
	By taking union bound over $1\leq k\neq j\leq K$, we can conclude that 
	\begin{align*}
		\Pr{\opnorm{\calP_{T_k}\calG\calA_{k,t}^\ast \calA_{j,t}\calG^\ast\calP_{T_j}} \leq \frac{\epsilon}{K} ,~\forall k\neq j} \geq 1-K(sn)^{-\gamma+1}.
	\end{align*}
If we choose $\gamma = \gamma'+\log K$, or $n\geq C_{\gamma'+\log K}\cdot \epsilon^{-2}K^2 \mu_0\mu_1sr\log(sn)$, then the probability of success is at least $1-(sn)^{-\gamma'+1}$.
\end{proof}

\begin{lemma}
	\label{lemma: T-RIP}
	Suppose $n\geq C_{\gamma} \epsilon^{-2}K^2\mu_0 s\mu_1r\log(sn)$. Then with probability at least $1-(sn)^{-\gamma+1}$, there holds the following inequality
	\begin{align*}
		\opnorm{\calP_T\tG(\calA^\ast \calA - \calI)\tG^\ast \calP_T} \leq \epsilon .
	\end{align*}
\end{lemma}
\begin{proof}
For any block diagonal matrix $\mZ = \blkdiag(\mZ_1, \cdots, \mZ_k)$ such that $\fronorm{\mZ}=1$, by the definitions of $\tG$ and $\calA$, one has
	\begin{align*}
		\la \mZ, \calP_T\tG(\calA^\ast \calA - \tcalI)\tG^\ast \calP_T(\mZ)\ra &=\sum_{\ell=1}^{K} \la \mZ_\ell , \calP_{T_\ell} \calG\calA_\ell^\ast \lb \sum_{k=1}^{K}\calA_k\calGT\calP_{T_k}(\mZ_k) \rb\ra - \sum_{\ell=1}^{K} \la \mZ_\ell,  \calP_{T_\ell}(\mZ_\ell)\ra\\
		&=\sum_{\ell=1}^{K} \la \mZ_\ell , \calP_{T_\ell} \calG\calA_\ell^\ast  \calA_\ell\calGT\calP_{T_\ell}(\mZ_\ell) \ra - \sum_{\ell=1}^{K} \la \mZ_\ell,  \calP_{T_\ell}(\mZ_\ell)\ra \\
		&\quad + \sum_{\ell=1}^K \la \calA_\ell\calGT\calP_{T_\ell}(\mZ_\ell), \sum_{k:k\neq \ell}\calA_k\calGT\calP_{T_k}(\mZ_k)\ra\\
		&=\underbrace{\sum_{\ell=1}^{K} \la \mZ_\ell , \calP_{T_\ell} \calG\calA_\ell^\ast  \calA_\ell\calGT\calP_{T_\ell}(\mZ_\ell) \ra - \sum_{\ell=1}^{K} \la \mZ_\ell,  \calP_{T_\ell}(\mZ_\ell)\ra}_{:=I_1} \\
		&\quad + \underbrace{\sum_{\ell=1}^K \lb \sum_{k:k\neq \ell} \la \calA_\ell\calGT\calP_{T_\ell}(\mZ_\ell), \calA_k\calGT\calP_{T_k}(\mZ_k)\ra \rb }_{:=I_2}.
	\end{align*}
\begin{itemize}
	\item Bounding of $I_1$. It is easy to see that
	\begin{align*}
		I_1 \leq \sum_{\ell=1}^{K} \opnorm{\calP_{T_\ell}\calG(\calA_\ell^\ast \calA_\ell - \calI)\calGT\calP_{T_\ell}}\cdot \fronorm{\mZ_\ell}^2\leq \frac{\epsilon}{K}\sum_{\ell=1}^{K}\fronorm{\mZ_\ell}^2
	\end{align*}
holds on the event $\calE_0$ with probability $\Pr{\calE_0} \geq 1-(sn)^{-\gamma+1}$
	\item Bounding of $I_2$. A direct calculation yields that
	\begin{align*}
		I_2 &=\sum_{\ell=1}^{K}\sum_{k:k\neq \ell} \la \calA_\ell\calGT\calP_{T_\ell}(\mZ_\ell), \calA_k\calGT\calP_{T_k}(\mZ_k)\ra\\
		&\leq \sum_{\ell=1}^{K}\sum_{k:k\neq \ell} \opnorm{\calP_{T_\ell}\calG\calA_\ell^\ast \calA_k\calGT\calP_{T_k}}\cdot \fronorm{\mZ_{\ell}}\cdot \fronorm{\mZ_k}\\
		&\leq \frac{\epsilon}{K}\sum_{k:k\neq \ell} \fronorm{\mZ_{\ell}}\cdot \fronorm{\mZ_k}
	\end{align*}
holds on the event $\calE_1$ with probability at least $1-(sn)^{-\gamma+1}$. 
\end{itemize}
Combining together all of the preceding bounds on $I_1$ and $I_2$, we have that
\begin{align*}
	\opnorm{\calP_T\tG(\calA^\ast \calA - \tcalI)\tG^\ast \calP_T}&=\sup_{\mZ=\blkdiag(\mZ_1,\cdots, \mZ_K):\fronorm{\mZ}=1}\la \mZ, \calP_T\tG(\calA^\ast \calA - \tcalI)\tG^\ast \calP_T(\mZ)\ra\\
	&\leq \frac{\epsilon}{K}\sum_{\ell=1}^{K}\fronorm{\mZ_\ell}^2 + \frac{\epsilon}{K}\sum_{\ell=1}^{K} \sum_{k:k\neq \ell} \fronorm{\mZ_{\ell}}\cdot \fronorm{\mZ_k}\\
	&=\frac{\epsilon}{K}\sum_{\ell=1}^{K} \fronorm{\mZ_\ell} \lb \fronorm{\mZ_{\ell}} +\sum_{k:k\neq \ell} \fronorm{\mZ_k} \rb\\
	&=\frac{\epsilon}{K} \lb \sum_{\ell=1}^{K} \fronorm{\mZ_\ell}\rb^2\\
	&\leq \epsilon \sum_{\ell=1}^{K} \fronorm{\mZ_\ell}^2\\
	&=\epsilon
\end{align*}
holds with probability at least $1-2(sn)^{-\gamma+1}$.
\end{proof}

\begin{lemma}
	\label{coro rip}
Suppose $n\geq C_{\gamma} \epsilon^{-2}K^2\mu_0 s\mu_1r\log(sn)$.  Then with probability at least $1-2(sn)^{-\gamma+1}$, there holds the following inequalities
	\begin{align*}
		\opnorm{\calA\tG^\ast\calP_T} &\leq \sqrt{1 +\epsilon},\\
		\opnorm{\calP_T\tG\lb \calA^\ast\calA - \calI \rb\tG^\ast } &\leq \sqrt{2K\mu_0s(1+\epsilon)}.
	\end{align*}
\end{lemma}
\begin{proof}
	
	For any block diagonal matrix $\mZ=\blkdiag(\mZ_1,\cdots, \mZ_K)$ such that $\fronorm{\mZ}=1$, one has
	\begin{align*}
		\fronorm{\calA\tG^\ast \calP_T(\mZ)}^2 &= \la\calA\tG^\ast \calP_T(\mZ), \calA\tG^\ast \calP_T(\mZ)\ra\\
		&= \la\calP_T\tG\calA^\ast\calA\tG^\ast \calP_T(\mZ), \mZ\ra\\
		&= \la\calP_T\tG \lb \calA^\ast\calA - \calI \rb\tG^\ast \calP_T(\mZ), \mZ\ra + \la \calP_T\tG\tG^\ast\calP_T(\mZ), \mZ\ra\\
		&\leq \opnorm{\calP_T\tG \lb \calA^\ast\calA - \calI \rb\tG^\ast \calP_T}\cdot \fronorm{\mZ}^2 + \fronorm{\mZ}^2\\
		&\leq (1+\epsilon) \fronorm{\mZ}^2.
	\end{align*}
Secondly, we bound the spectral norm of $\calP_T\tG\lb \calA^\ast\calA - \calI \rb\tG^\ast$. A direct computation yields that
	\begin{align*}
	\opnorm{\calP_T\tG\lb \calA^\ast\calA - \calI \rb\tG^\ast } &\leq \opnorm{\calP_T\tG\lb \calA^\ast\calA - \calI \rb } \\
	&\leq \opnorm{\calP_T\tG  \calA^\ast\calA  } +1\\
	&\leq \opnorm{\calP_T\tG  \calA^\ast} \cdot \opnorm{\calA}+1\\
	&\leq \sqrt{K\mu_0s(1+\epsilon)}+1\\
	&\leq \sqrt{2K\mu_0s(1+\epsilon)},
\end{align*}
where we have used the fact that $\opnorm{\tG^\ast}\leq 1$ and $\opnorm{\calA}\leq \sqrt{K\mu_0s}$.The fact $\opnorm{\calA}\leq \sqrt{K\mu_0s}$  can be proved as follows:
for any block diagonal matrix $\mZ=\blkdiag(\mZ_1,\cdots, \mZ_K)$, one has
\begin{align*}
	\opnorm{\calA\tG(\mZ)} &= \fronorm{\sum_{k=1}^{K}\calA_k\calG(\mZ_k)} \leq \sum_{k=1}^{K} \sqrt{\mu_0s}\fronorm{\mZ_k} \leq \sqrt{K\mu_0s}\fronorm{\mZ}.
\end{align*}
Thus we complete the proof. 
\end{proof}

\begin{lemma}
	\label{lemma inner product}
	Let $\mPhi_1, \mPhi_2, \mPsi_1, \mPsi_2$ be block diagonal matrices such that $\mPhi_1,\mPhi_2\in T $ and $\mPsi_1, \mPsi_2\in T^\perp$. Suppose $n\geq C_{\gamma} \epsilon^{-2}K^2\mu_0 s\mu_1r\log(sn)$.  Then with probability at least $1-2(sn)^{-\gamma+1}$, there holds the following inequality
	\begin{align*}
		\left| \la \tG \lb \calA^\ast\calA-\calI \rb\tG^\ast(\mPhi_1 + \mPsi_1), \mPhi_2+\mPsi_2\ra\right| &\leq \epsilon \cdot \fronorm{\mPhi_1} \cdot  \fronorm{\mPhi_2} \\
		&\quad + \sqrt{2K\mu_0s(1+\epsilon)}\cdot \lb \fronorm{\mPhi_1} \cdot \fronorm{\mPsi_2} +\fronorm{ \mPsi_1}\cdot \fronorm{\mPhi_2} \rb \\
		&\quad + 2K\mu_0s \fronorm{\mPsi_1}\cdot \fronorm{\mPsi_2}.
	\end{align*}
\end{lemma}
\begin{proof}
	The triangle inequality gives that
	\begin{align*}
		\left| \la \tG \lb \calA^\ast\calA-\calI \rb\tG^\ast(\mPhi_1 + \mPsi_1), \mPhi_2+\mPsi_2\ra\right| 
		&\leq 	\left| \la \tG \lb \calA^\ast\calA-\calI \rb\tG^\ast(\mPhi_1), \mPhi_2\ra\right|+\left| \la \tG \lb \calA^\ast\calA-\calI \rb\tG^\ast(\mPhi_1), \mPsi_2\ra\right|\\
		&\quad +\left| \la \tG \lb \calA^\ast\calA-\calI \rb\tG^\ast(\mPsi_1), \mPhi_2\ra\right|+\left| \la \tG \lb \calA^\ast\calA-\calI \rb\tG^\ast(\mPsi_1), \mPsi_2\ra\right|\\
		&=\underbrace{\left| \la \tG \lb \calA^\ast\calA-\calI \rb\tG^\ast\calP_T(\mPhi_1), \calP_T(\mPhi_2)\ra\right|}_{:=I_3}\\
		&\quad +\underbrace{\left| \la \tG \lb \calA^\ast\calA-\calI \rb\tG^\ast\calP_T(\mPhi_1), \mPsi_2\ra\right|+\left| \la \tG \lb \calA^\ast\calA-\calI \rb\tG^\ast(\mPsi_1), \calP_T(\mPhi_2)\ra\right|}_{:=I_4}\\
		&\quad +\underbrace{\left| \la \tG \lb \calA^\ast\calA-\calI \rb\tG^\ast(\mPsi_1), \mPsi_2\ra\right|}_{:=I_5}.
	\end{align*}
	\begin{itemize}
		\item Bounding of $I_3$. A direct computation yields that
		\begin{align*}
			I_1 &=\left| \la \tG \lb \calA^\ast\calA-\calI \rb\tG^\ast\calP_T(\mPhi_1), \calP_T(\mPhi_2)\ra\right|\\
			&\leq \opnorm{ \calP_T \tG \lb \calA^\ast\calA-\calI \rb\tG^\ast\calP_T}\cdot \fronorm{\mPhi_1}\cdot \fronorm{\mPhi_2}\\
			&\leq\epsilon\cdot \fronorm{\mPhi_1} \cdot  \fronorm{\mPhi_2},
		\end{align*}
		where the last line is due to Lemma \ref{lemma: T-RIP}.
		\item Bounding of $I_4$. It is easy to see that 
		\begin{align*}
			I_2 &\leq \left| \la \tG \lb \calA^\ast\calA-\calI \rb\tG^\ast\calP_T(\mPhi_1), \mPsi_2\ra\right|+\left| \la \tG \lb \calA^\ast\calA-\calI \rb\tG^\ast(\mPsi_1), \calP_T(\mPhi_2)\ra\right| \\
			&\leq \opnorm{\tG \lb \calA^\ast\calA-\calI \rb\tG^\ast\calP_T}\cdot \fronorm{\mPhi} \cdot \fronorm{\mPsi_2} + \opnorm{\tG \lb \calA^\ast\calA-\calI \rb\tG^\ast\calP_T}\cdot \fronorm{ \mPsi_1}\cdot \fronorm{\mPhi_2}\\
			&\leq \sqrt{2K\mu_0s(1+\epsilon)}\cdot \lb \fronorm{\mPhi_1} \cdot \fronorm{\mPsi_2}  +\fronorm{ \mPsi_1}\cdot \fronorm{\mPhi_2} \rb,
		\end{align*}
		where the last line is due to Lemma \ref{coro rip}.
		\item Bounding of $I_5$. A straightforward computation shows that
		\begin{align*}
			I_3 &=\left| \la \tG \lb \calA^\ast\calA-\calI \rb\tG^\ast(\mPsi_1), \mPsi_2\ra\right|\\
			&\leq \opnorm{\tG \lb \calA^\ast\calA-\calI \rb\tG^\ast} \cdot \fronorm{\mPsi_1}\cdot \fronorm{\mPsi_2}\\
			&\leq \lb \opnorm{\calA}^2 + 1\rb\fronorm{\mPsi_1}\cdot \fronorm{\mPsi_2}\\
			&\leq 2K\mu_0s \fronorm{\mPsi_1}\cdot \fronorm{\mPsi_2}.
		\end{align*}
	\end{itemize}
	Combing together, one has
	\begin{align*}
		\left| \la \tG \lb \calA^\ast\calA-\calI \rb\tG^\ast(\mPhi_1 + \mPsi_1), \mPhi_2+\mPsi_2\ra\right|  &\leq\epsilon \cdot \fronorm{\mPhi_1} \cdot  \fronorm{\mPhi_2} \\
		&\quad + \sqrt{2K\mu_0s(1+\epsilon)}\cdot \lb \fronorm{\mPhi_1} \cdot \fronorm{\mPsi_2} +\fronorm{ \mPsi_1}\cdot \fronorm{\mPhi_2} \rb \\
		&\quad + 2K\mu_0s \fronorm{\mPsi_1}\cdot \fronorm{\mPsi_2},
	\end{align*}
which completes the proof.
\end{proof}

\begin{lemma}
	\label{lemma a}
Let the SVD of the matrix $\mZ_{k,\natural}$ be $\mZ_{k, \natural}=\mU_{k,\natural}\bSigma_{k,\natural}\mV_{k,\natural}^\tranH$. Denote $\mL_{k,\natural} = \mU_{k,\natural}\bSigma_{k,\natural}^{\frac{1}{2}}, \mR_{k,\natural} = \mV_{k,\natural}\bSigma_{k,\natural}^{\frac{1}{2}}, \mDelta_k = \mL_k\mO_k - \mL_{k, \natural} $ and $\mUpsilon_k=\mR_k\mO_k-\mR_{k,\natural}$, where $\mO_k = \arg\min_{\mO^\tranH\mO=\mO\mO^\tranH=\mI_r} \sqrt{ \fronorm{\mL_k\mO - \mL_{k,\natural}}^2 + \fronorm{\mR_k\mO - \mR_{k,\natural}}^2}$. If  it holds that $\fronorm{\mL_k\mR_k^\tranH - \mZ_{k, \natural}} \leq \frac{\delta}{\sqrt{K}}\sigma_r(\mZ_{k, \natural})$ for all $k\in [K]$ and $\delta <\frac{1}{2\kappa}$, then the following inequalities are satisfied:
	\begin{align*}
		\max\lb \opnorm{\mDelta\bSigma_{\natural}^{-\frac{1}{2}}} , \opnorm{\mUpsilon\bSigma_{\natural}^{-\frac{1}{2}}} \rb  &\leq \frac{2\kappa\delta }{\sqrt{K}},\\
		\opnorm{\mDelta(\mL^\tranH\mL)^{-\frac{1}{2}}} &\leq \frac{1}{\sqrt{K}}\frac{2\kappa\delta}{1-2\kappa\delta},\\ \opnorm{\mUpsilon(\mR^\tranH\mR)^{-\frac{1}{2}}} &\leq \frac{1}{\sqrt{K}}\frac{2\kappa\delta}{1-2\kappa\delta},
	\end{align*}
where $\mL, \mR, \bSigma_{ \natural}, \mDelta, \mUpsilon$ are block diagonal matrices constructed from the corresponding matrices $\mL_k, \mR_k, \bSigma_{k,\natural}, \mDelta_k, \mUpsilon_k$, respectivel, and $\kappa = \frac{\max_k \sigma_1(\mZ_{k, \natural})}{\min_k\sigma_r(\mZ_{k, \natural})}$.
\end{lemma}
\begin{proof}
	Notice that $\mDelta\bSigma_{\natural}^{-\frac{1}{2}}$ is a block diagonal matrix. A direct computation yields that 
	\begin{align*}
		\opnorm{\mDelta\bSigma_{\natural}^{-\frac{1}{2}}} &=\max_k \opnorm{\mDelta_k\bSigma_{k,\natural}^{-\frac{1}{2}}}\\
		&\leq \max_k \lb \opnorm{\mDelta_k\bSigma_{k,\natural}^{\frac{1}{2}}} \cdot \opnorm{\bSigma_{k,\natural}^{-1}}\rb\\
		&=\max_k \lb \opnorm{\mDelta_k\bSigma_{k,\natural}^{\frac{1}{2}}} \cdot\frac{1}{\sigma_r(\mZ_{k, \natural})}\rb\\
		&\leq \frac{1}{\min_k\sigma_r(\mZ_{k, \natural})}\cdot \max_k  \opnorm{\mDelta_k\bSigma_{k,\natural}^{\frac{1}{2}}} \\
		&\stackrel{(a)}{\leq}\frac{1}{\min_k\sigma_r(\mZ_{k, \natural})}\cdot \max_k  \lb \sqrt{\sqrt{2}+1}\fronorm{\mL_k\mR_k^\tranH - \mZ_{k, \natural}}\rb\\
		&\leq \frac{2}{\min_k\sigma_r(\mZ_{k, \natural})}\cdot \frac{\delta}{\sqrt{K}}\max_k\sigma_r(\mZ_{k,\natural})\\
		&\leq \frac{2}{\min_k\sigma_r(\mZ_{k, \natural})}\cdot \frac{\delta}{\sqrt{K}}\max_k\sigma_1(\mZ_{k,\natural})\\
		&=\frac{2\kappa\delta}{\sqrt{K}},
	\end{align*}
	where step (a) is due to Lemma \ref{lemma 24 cong}.  Moreover, a direct computation yields that 
	\begin{align*}
		\opnorm{\mDelta(\mL^\tranH\mL)^{-\frac{1}{2}}} &=\opnorm{\mDelta\bSigma_{\natural}^{-\frac{1}{2}}\bSigma_{\natural}^{\frac{1}{2}}(\mL^\tranH\mL)^{-\frac{1}{2}}}\\
		&\leq \opnorm{\mDelta\bSigma_{\natural}^{-\frac{1}{2}}}\cdot \opnorm{\bSigma_{\natural}^{\frac{1}{2}}(\mL^\tranH\mL)^{-\frac{1}{2}}}\\
		&\leq \max_k \opnorm{(\mL_k^\tranH\mL_k)^{-\frac{1}{2}} \bSigma_{k,\natural}^{\frac{1}{2}} }\cdot \frac{2\kappa\delta}{\sqrt{K}}\\
		&\stackrel{(a)}{\leq } \frac{1}{1-\min_k \opnorm{\mDelta_k\bSigma_{k, \natural}^{-\frac{1}{2}}}}\cdot \frac{2\kappa\delta}{\sqrt{K}}\\
		&\leq \frac{1}{1-\opnorm{\mDelta\bSigma_{ \natural}^{-\frac{1}{2}}}}\cdot \frac{2\kappa\delta}{\sqrt{K}}\\
		&\leq \frac{1}{\sqrt{K}}\frac{2\kappa \delta}{1-\frac{2\kappa \delta}{\sqrt{K}}}\\
		&\leq \frac{1}{\sqrt{K}}\frac{2\kappa \delta}{1-2\kappa \delta},
	\end{align*}
	where step (a) follows from Lemma \ref{lemma 25 cong}. 	Using the same argument, one can show that 
	\begin{align*}
		\opnorm{\mUpsilon\bSigma_{\natural}^{-\frac{1}{2}}} &\leq  \frac{2\kappa\delta}{\sqrt{K}},\\
		\opnorm{\mUpsilon(\mR^\tranH\mR)^{-\frac{1}{2}}} &\leq \frac{1}{\sqrt{K}}\frac{2\kappa\delta}{1-2\kappa\delta}.
	\end{align*}
	Thus we complete the proof.
\end{proof}

\begin{lemma}
	\label{lemma phi and psi}
Let $\kappa  = \frac{\max_k \sigma_1(\mZ_{k, \natural})}{\min_k \sigma_r(\mZ_{k, \natural})}$.  Suppose $\fronorm{\mL_k\mR_k^\tranH - \mZ_{k, \natural}} \leq \frac{\delta}{\sqrt{K}}\sigma_r(\mZ_{k, \natural})$ for any $k\in[K]$, where $\delta$ is a constant such that $\delta < \frac{1}{6\kappa }$. Then one has
\begin{align*}
		\fronorm{\mDelta\mUpsilon^\tranH} &\leq \frac{6\kappa \delta}{\sqrt{K}}   \fronorm{\mL\mR^\tranH - \mZ_{\natural}},\\
	\lb 1- \frac{6\kappa\delta}{\sqrt{K}}\rb \fronorm{\mL\mR^\tranH - \mZ_{\natural}}&\leq\fronorm{\mL_{\natural}\mUpsilon^\tranH + \mDelta\mR_{\natural}^\tranH } \leq \lb 1+ \frac{6\kappa\delta}{\sqrt{K}}\rb \fronorm{\mL\mR^\tranH - \mZ_{\natural}}.
\end{align*}
\end{lemma}
\begin{proof}
Firstly, a direct computation yields that 
\begin{align*}
\fronorm{\mDelta\mUpsilon^\tranH} =\fronorm{\mDelta\bSigma_{\natural}^{-\frac{1}{2}} \bSigma_{\natural}^{\frac{1}{2}}\mUpsilon^\tranH}
	\leq \opnorm{\mDelta\bSigma_{\natural}^{-\frac{1}{2}} }  \fronorm{\mUpsilon \bSigma_{\natural}^{\frac{1}{2}}}
	\stackrel{(a)}{\leq } \frac{2\kappa \delta}{\sqrt{K}} \fronorm{\mUpsilon \bSigma_{\natural}^{\frac{1}{2}}},
\end{align*}
where step (a) follows from Lemma \ref{lemma a}. On the other hand, one can also obtain that 
\begin{align*}
\fronorm{\mDelta\mUpsilon^\tranH}=\fronorm{\mDelta\bSigma_{\natural}^{\frac{1}{2}} \bSigma_{\natural}^{-\frac{1}{2}}\mUpsilon^\tranH} \leq \fronorm{\mDelta\bSigma_{\natural}^{\frac{1}{2}} }  \opnorm{\mUpsilon \bSigma_{\natural}^{-\frac{1}{2}}}	\stackrel{(a)}{\leq } \frac{2\kappa\delta}{\sqrt{K}} \fronorm{\mDelta \bSigma_{\natural}^{\frac{1}{2}}},
\end{align*}
where last line is due to Lemma \ref{lemma a}. Combining together, one has
\begin{align*}
\fronorm{\mDelta\mUpsilon^\tranH}&\leq \frac{2\kappa\delta}{\sqrt{K}} \lb \fronorm{\mDelta \bSigma_{\natural}^{\frac{1}{2}}} +\fronorm{\mUpsilon \bSigma_{\natural}^{\frac{1}{2}}} \rb\\
	&\leq  \frac{2\sqrt{2}\kappa\delta}{\sqrt{K}} \lb \fronorm{\mDelta \bSigma_{\natural}^{\frac{1}{2}}}^2 +\fronorm{\mUpsilon \bSigma_{\natural}^{\frac{1}{2}}}^2 \rb^{\frac{1}{2}}\\
	&=\frac{2\sqrt{2}\kappa\delta}{\sqrt{K}} \lb \sum_{k=1}^{K}\lb \fronorm{\mDelta_k \bSigma_{k,\natural}^{\frac{1}{2}}}^2 +\fronorm{\mUpsilon_k \bSigma_{k,\natural}^{\frac{1}{2}}}^2 \rb\rb^{\frac{1}{2}}\\
	&\leq \frac{2\sqrt{2}\kappa\delta}{\sqrt{K}} \lb \sqrt{2}+1 \rb^{\frac{1}{2}} \fronorm{\mL\mR^\tranH - \mZ_{\natural}}\\
	&\leq \frac{6\kappa \delta}{\sqrt{K}}   \fronorm{\mL\mR^\tranH - \mZ_{\natural}},
\end{align*}
where the fourth line is due to Lemma \ref{lemma 24 cong}. 
Secondly, notice that 
\begin{align*}
	\mL\mR^\tranH - \mZ_{\natural} &= (\mL_{\natural} + \mDelta)(\mR_{\natural} + \mUpsilon)^\tranH - \mZ_{\natural} =\mL_{\natural}\mUpsilon^\tranH + \mDelta\mR_{\natural}^\tranH + \mDelta\mUpsilon^\tranH.
\end{align*}
Thus a simple triangular inequality yields that
\begin{align*}
	\fronorm{\mL_{\natural}\mUpsilon^\tranH + \mDelta\mR_{\natural}^\tranH } &= \fronorm{\mL\mR^\tranH -\mZ_{\natural} - \mDelta\mUpsilon^\tranH }\\
	&\leq \fronorm{\mL\mR^\tranH -\mZ_{\natural} }+\fronorm{ \mDelta\mUpsilon^\tranH }\\
	&\leq \lb 1+\frac{6\kappa\delta}{\sqrt{K}}\rb \fronorm{\mL\mR^\tranH - \mZ_{\natural}}.
\end{align*}
On the other hand, one has
\begin{align*}
	\fronorm{\mL_{\natural}\mUpsilon^\tranH + \mDelta\mR_{\natural}^\tranH } 	&\geq \fronorm{\mL\mR^\tranH -\mZ_{\natural} }-\fronorm{ \mDelta\mUpsilon^\tranH }\\
	&\geq \lb 1-\frac{6\kappa\delta}{\sqrt{K}}\rb \fronorm{\mL\mR^\tranH - \mZ_{\natural}}.
\end{align*}
Thus we complete the proof.
\end{proof}

\begin{lemma}
	\label{fact 4}
	Denote $\kappa = \frac{\max_k\sigma_1(\mZ_{k, \natural})}{\min_k\sigma_r(\mZ_{k, \natural})}$. Suppose $\fronorm{\mL_k\mR_k^\tranH - \mZ_{k, \natural}} \leq \frac{\delta}{\sqrt{K}}\sigma_r(\mZ_{k, \natural})$ for any $k\in[K]$, where $\delta$ is a constant such that $\delta<\frac{1}{2\kappa}$. Then one has
	\begin{align*}
		\opnorm{\mL_{\natural}(\mL^\tranH\mL)^{-\frac{1}{2}}} &\leq \frac{1}{1-2\kappa\delta},\\
		\opnorm{\mR_{\natural}(\mR^\tranH\mR)^{-\frac{1}{2}}}&\leq \frac{1}{1-2\kappa\delta},\\
		\opnorm{\mDelta(\mL^\tranH\mL)^{-\frac{1}{2}}} &\leq \frac{1}{\sqrt{K}} \frac{2\kappa\delta}{1-2\kappa\delta},\\
		\opnorm{\mUpsilon(\mR^\tranH\mR)^{-\frac{1}{2}}}&\leq \frac{1}{\sqrt{K}}\frac{2\kappa\delta}{1-2\kappa\delta}
	\end{align*}
\end{lemma}
\begin{proof}
Recall that $\mL_{\natural} =\mU_{\natural}\bSigma_{\natural}^{\frac{1}{2}}= \blkdiag(\mU_{1,\natural}\bSigma_{1,\natural}^{\frac{1}{2}}, \cdots, \mU_{K,\natural}\bSigma_{K,\natural}^{\frac{1}{2}})$ and $\mL=\blkdiag(\mL_1,\cdots, \mL_K)$. A direct computation yields that 
\begin{align*}
	\opnorm{\mL_{\natural}(\mL^\tranH\mL)^{-\frac{1}{2}}} &=\opnorm{\bSigma_{\natural}^{\frac{1}{2}}(\mL^\tranH\mL)^{-\frac{1}{2}}} =\max_k \opnorm{\bSigma_{k,\natural}^{\frac{1}{2}}(\mL_k^\tranH\mL_k)^{-\frac{1}{2}}}
	=\max_k \opnorm{\mL_k(\mL_k^\tranH\mL_k)^{-1}\bSigma_{k,\natural}^{\frac{1}{2}}}\\
	&\stackrel{(a)}{\leq} \max_k\frac{1}{1- \opnorm{\mDelta_k\bSigma_{k,\natural}^{-\frac{1}{2}}}}=\frac{1}{1- \max_k\opnorm{\mDelta_k\bSigma_{k,\natural}^{-\frac{1}{2}}}}
	=\frac{1}{1- \opnorm{\mDelta\bSigma_{\natural}^{-\frac{1}{2}}}}\\
	&\stackrel{(b)}{\leq }\frac{1}{1-\frac{2\kappa\delta}{\sqrt{K}}}\\
	&\leq \frac{1}{1-2\kappa \delta},
\end{align*}
where step (a) follows from Lemma \ref{lemma 25 cong} and step (b) is due to Lemma \ref{lemma a}. Moreover, one has
\begin{align*}
	\opnorm{\mDelta(\mL^\tranH\mL)^{-\frac{1}{2}}} &= \opnorm{\mDelta\bSigma_{\natural}^{-\frac{1}{2}}\bSigma_{\natural}^{\frac{1}{2}}(\mL^\tranH\mL)^{-\frac{1}{2}}}\\
	&\leq \opnorm{\mDelta\bSigma_{\natural}^{-\frac{1}{2}}}\cdot \opnorm{\bSigma_{\natural}^{\frac{1}{2}}(\mL^\tranH\mL)^{-\frac{1}{2}}}\\
	&\leq \frac{2\kappa\delta}{\sqrt{K}}\cdot \frac{1}{1-2\kappa\delta},
\end{align*}
where the last line is due to Lemma \ref{lemma a}.  Using the same argument, one can show that 
\begin{align*}
	\opnorm{\mR_{\natural}(\mR^\tranH\mR)^{-\frac{1}{2}}}&\leq \frac{1}{1-2\kappa \delta},\\
	\opnorm{\mUpsilon(\mR^\tranH\mR)^{-\frac{1}{2}}} &\leq \frac{1}{\sqrt{K}} \frac{2\kappa\delta}{1-2\kappa\delta}.
\end{align*}

\end{proof}

\section{Auxiliary Lemmas}

\begin{lemma}\cite[Lemma VI.3]{mao2022blind}
\label{lemma initial 1}
Suppose that $\mZ_{k,\natural}$ is $\mu_1$-incoherence. Then with probability at least $1-(sn)^{-c_1}$
\begin{align*}
	\opnorm{\calG\calA^\ast_k\calA_k\calG^\ast(\mZ_{k,\natural})- \mZ_{k,\natural}} \leq c_0 \sqrt{\frac{\mu_0s \mu_1 r\log(sn)}{n}}\sigma_1(\mZ_{k,\natural}),
\end{align*}
where $c_0>0, c_1\geq 1$ are absolute constant. 
\end{lemma}

\begin{lemma}\cite[Lemma VI.6]{mao2022blind}
	\label{lemma v6}
Suppose $\mZ_{k,\natural}$ is $\mu_1$-incoherence. Then one has 
\begin{align*}
	\sqrt{\sum_{i=1}^{n} \frac{\twonorm{\calGT(\mZ_{k,\natural}\ve_i)}^2}{w_i}} &\leq c_1 \sqrt{ \frac{\mu_1 r\log(sn)}{n}} \sigma_1(\mZ_{k,\natural}),\\
	\max_{1\leq i\leq n} \frac{\twonorm{\calGT(\mZ_{k,\natural})\ve_i}}{\sqrt{w_i}}& \leq \frac{\mu_1 r}{n} \sigma_1(\mZ_{k,\natural}).
\end{align*}
\end{lemma}

\begin{lemma}\cite[Lemma 24]{tong2021accelerating}
	\label{lemma 24 cong}
	Given matrices $\mL_k,\mR_k$, let $\mO_k = \arg\min_{\mO^\tranH\mO=\mO\mO^\tranH=\mI_r} \sqrt{ \fronorm{\mL_k\mO - \mL_{k,\natural}}^2 + \fronorm{\mR_k\mO - \mR_{k,\natural}}^2}$. Then one has
	\begin{align*}
		\fronorm{(\mL_k\mO_k - \mL_{k,\natural})\bSigma_{k,\natural}^{\frac{1}{2}}}^2 + \fronorm{(\mR_k\mO_k - \mR_{k,\natural})\bSigma_{k,\natural}^{\frac{1}{2}}}^2 \leq (\sqrt{2}+1) \fronorm{\mL_k\mR_k^\tranH - \mZ_{k, \natural}}^2.
	\end{align*}
\end{lemma}

\begin{lemma}[Lemma 25 in Cong]
	\label{lemma 25 cong}
	For any matrix $\mL_k, \mR_k$, denote $\mDelta_k = \mL_k - \mL_{k,\natural}$ and $\mUpsilon_{k} = \mR_k - \mR_{k,\natural}$. Suppose $\max \lb \opnorm{\mDelta_k\bSigma_{k,\natural}^{-\frac{1}{2}}}, \opnorm{\mUpsilon_{k}\bSigma_{k,\natural}^{-\frac{1}{2}}}\rb<1$, then one has
	\begin{align*}
		\opnorm{\mL_k(\mL_k^\tranH\mL_k)^{-1}\bSigma_{k,\natural}^{\frac{1}{2}}} &\leq \frac{1}{1 -\opnorm{\mDelta_k\bSigma_{k,\natural}^{-\frac{1}{2}}} },\\
		\opnorm{\mR_k(\mR_k^\tranH\mR_k)^{-1}\bSigma_{k,\natural}^{\frac{1}{2}}}&\leq \frac{1}{1 -\opnorm{\mUpsilon_k\bSigma_{k,\natural}^{-\frac{1}{2}}} }.
	\end{align*}
\end{lemma}

\bibliographystyle{unsrt}
\bibliography{refs}
\end{document}